\documentclass[12pt,onecolumn,draftclsnofoot]{IEEEtran}
\usepackage{amsmath,amssymb,amsthm,bm}
\usepackage[dvipsnames]{xcolor}
\usepackage[hidelinks]{hyperref}
\usepackage{aliascnt}
\usepackage{booktabs}
\usepackage{tabularx}
\usepackage{array}
\newtheorem{theorem}{Theorem}
\newaliascnt{lemma}{theorem}
\newtheorem{lemma}[lemma]{Lemma}
\aliascntresetthe{lemma}
\newaliascnt{proposition}{theorem}
\newtheorem{proposition}[proposition]{Proposition}
\aliascntresetthe{proposition}
\newaliascnt{corollary}{theorem}
\newtheorem{corollary}[corollary]{Corollary}
\aliascntresetthe{corollary}
\newaliascnt{assumption}{theorem}
\newtheorem{assumption}[assumption]{Assumption}
\aliascntresetthe{assumption}
\usepackage[nameinlink]{cleveref}
\crefname{theorem}{Theorem}{Theorems}
\Crefname{theorem}{Theorem}{Theorems}
\crefname{lemma}{Lemma}{Lemmas}
\Crefname{lemma}{Lemma}{Lemmas}
\crefname{proposition}{Proposition}{Propositions}
\Crefname{proposition}{Proposition}{Propositions}
\crefname{corollary}{Corollary}{Corollaries}
\Crefname{corollary}{Corollary}{Corollaries}
\crefname{assumption}{Assumption}{Assumptions}
\Crefname{assumption}{Assumption}{Assumptions}
\crefname{section}{Section}{Sections}
\Crefname{section}{Section}{Sections}
\crefname{table}{Table}{Tables}
\Crefname{table}{Table}{Tables}
\crefname{figure}{Figure}{Figures}
\Crefname{figure}{Figure}{Figures}
\crefformat{equation}{#2(#1)#3}
\Crefformat{equation}{#2(#1)#3}
\crefrangeformat{equation}{#3(#1)#4--#5(#2)#6}
\Crefrangeformat{equation}{#3(#1)#4--#5(#2)#6}
\crefmultiformat{equation}%
  {#2(#1)#3}{, #2(#1)#3}{, #2(#1)#3}{, #2(#1)#3}
\Crefmultiformat{equation}%
  {#2(#1)#3}{, #2(#1)#3}{, #2(#1)#3}{, #2(#1)#3}
\newcommand{\tang}{\mathcal{T}_\star}
\newcommand{\tangp}{\tang^\perp}
\newcommand{\proj}[1]{P_{#1}}
\newcommand{\HT}{\bm{H}_{\tang}}
\newcommand{\HTp}{\bm{H}_{\tangp}}
\newcommand{\Gamnm}{\log^2(2n)\log^2(2m)}
\newcommand{\Lameta}{\Gamnm + \log\frac{2}{\eta}}
\newcommand{\ueta}{\log\frac{4m}{\eta}}
\newcommand{\mustar}{4M^2\ueta}

\title{Stable Recovery and Benign Overparameterized Landscapes for Phase Retrieval from Coded Diffraction Patterns}
\author{Jian-Feng Cai, Zhibo Jin, Tong Wu, and Ruizhe Xia%
\thanks{Jian-Feng Cai, Zhibo Jin, and Tong Wu are with The Hong Kong
University of Science and Technology, Hong Kong SAR, China (e-mail:
jfcai@ust.hk; zjinay@connect.ust.hk; twubi@connect.ust.hk).}%
\thanks{Ruizhe Xia is with the Department of Applied Mathematics, The Hong
Kong Polytechnic University, Hong Kong (e-mail: ruizhxia@polyu.edu.hk).}}
\begin{document}
\maketitle
\begin{abstract}
Coded diffraction patterns (CDPs) provide a structured and physically relevant
model for phase retrieval, but the dependence among Fourier measurements
generated by a common mask makes sharp stability analysis challenging. For a
fixed unit-norm signal \(\bm x_\star \in \mathbb{C}^n\), let
\(\bm X_\star=\bm x_\star\bm x_\star^*\), and let
\(\mathcal A\) be the lifted CDP measurement operator.
We prove that, with \(L=O(\log n)\) random masks, the following uniform
lower isometry holds with high probability:
\[
\|\bm X-\bm X_\star\|_F
\lesssim
\log^2(2n)
\frac{\|\mathcal A(\bm X-\bm X_\star)\|_2}{\sqrt{nL}},
\qquad
\bm X\succeq\bm 0,
\]
from which we derive two consequences. First, for
\(\bm y=\mathcal A(\bm X_\star)+\bm e\), PhaseLift-type convex programs
achieve the Gaussian-type stable recovery bound
$
\|\widehat{\bm X}-\bm X_\star\|_F
\lesssim
\frac{\log^2(2n)}{\sqrt{nL}}\|\bm e\|_2.
$
Second, in the noiseless case, the nonconvex factorized loss has a benign landscape
when the factor width satisfies \(r = O(\log^5(2n))\): every second-order
critical point \(\bm V\in\mathbb C^{n\times r}\) satisfies
\(\bm V\bm V^*=\bm X_\star\).
The key ingredient is a uniform operator-norm bound over row subsets
of the dependent CDP measurement matrix, which permits the removal of a
controlled set of adaptively selected rows while preserving tangent injectivity.
\end{abstract}

\section{Introduction}
\label{sec:intro}

Phase retrieval concerns the recovery of a signal from magnitude-only
measurements. Given an unknown signal \(\bm x_\star\in\mathbb C^n\), one
observes
\begin{equation} \label{eq: phase retrieval model}
    y_i
    =
    \left|\left\langle \bm a_i,\bm x_\star\right\rangle\right|^2+e_i,
    \qquad i=1,\ldots,m,
\end{equation}
where \(\bm a_i\in\mathbb C^n\) are known measurement vectors and \(e_i\)
denotes additive measurement noise.
The loss of phase makes the inverse problem nonlinear and, without
sufficient measurement diversity, ill-posed~\cite{beinert2015ambiguities,bendory2017fourier}.
Phase retrieval arises in
diffraction imaging, crystallography, microscopy, astronomy, and other
imaging applications~\cite{dainty1987astronomy,harrison1993phase,millane1990phase,
shechtman2015overview,jaganathan2016overview}.

Independent Gaussian measurements are standard benchmarks for phase
retrieval, with guarantees for convex lifting methods such as
PhaseLift~\cite{candes2013phaselift,candes2014quadratic}
and for nonconvex algorithms~\cite{candes2015wirtinger,chen2017solving,
chen2019gradient,duchi2019solving,godeme2023provable,tan2023online,
wang2018truncated,zhang2017nonconvex}.
Beyond analyzing algorithmic dynamics, an elegant alternative is to establish a \emph{benign landscape} for the loss: every local minimizer is global.
Such geometry offers an algorithm-independent explanation for why local-search methods can reach global solutions without carefully designed
initialization~\cite{sun2018geometric,chen2019gradient,tan2023online}.
For the nonconvex quartic intensity loss over \(\mathbb C^{n\times 1}\),
a benign landscape was first established in~\cite{sun2018geometric}
and later shown to hold with \(m=O(n\log n)\)
in~\cite{cai2023nearly}.\footnote{Although the theorem is stated for the
real-valued setting, Cai et al.~\cite[Remark~3]{cai2023nearly} explain how the
result extends to the complex-valued setting. We therefore do not distinguish
between the real and complex settings in our landscape discussion.}
At \(m=O(n)\), this landscape seems unfavorable~\cite{liu2024local}, but~\cite{mcrae2026benign} shows that optimizing over \(\mathbb C^{n\times r}\) with \(r=O(\log n)\) can guarantee a benign landscape.
Such \emph{overparameterization} bridges nonconvex and convex recovery: it trades greater computational cost for more descent directions, achieving a more convex-like landscape; see \cref{sec:overparameterized-result}.
We note that guarantees for phase retrieval with sub-Gaussian and heavy-tailed measurements have also been developed~\cite{krahmer2020complex,huangli2026heavy,mcrae2026benign}.

However, physical acquisition systems produce structured and correlated measurement vectors in many applications, for which row-independent models are not representative.
Phase retrieval from \emph{coded diffraction pattern} (CDP) measurements is an
important model that strikes a balance between physical
practicality and theoretical tractability.
It uses the Fourier transform to model data acquisition from an object under the Fraunhofer diffraction approximation~\cite{loewen2018diffraction,gerchberg1972practical,fienup1982phase,fannjiang2020numerics}, while masks modulate the signal to increase measurement diversity and improve recovery guarantees~\cite{fannjiang2012phase,candes2015cdp}.
Formally, the CDP model takes the form
\begin{equation} \label{eq: CDP-model}
y_{\ell k}
=
\left|
\left\langle \bm f_k,\bm D_\ell\bm x_\star\right\rangle
\right|^2
+e_{\ell k},
\qquad
1\leq \ell\leq L,\quad 0\leq k<n,
\end{equation}
where \(\bm D_\ell=\operatorname{diag}(\bm d_\ell)\) is the \(\ell\)-th mask,
\(\{\bm f_k\}_{k=0}^{n-1}\) are the columns of the unnormalized
\(n\times n\) discrete Fourier transform matrix, and \(e_{\ell k}\) denotes measurement noise.
Each mask produces one diffraction pattern from \(n\) Fourier-intensity
measurements, which are generally correlated within each pattern;
\(L\) masks produce \(m=nL\) measurements in total.
Unless otherwise specified, we consider CDPs under the standard random-mask
model described in \cref{sec:preliminaries}; throughout, \(\bm x_\star\) is a
fixed unit vector independent of the masks.

Guarantees for exact recovery from noiseless CDPs have been progressively sharpened to
the order-optimal complexity $L=O(\log n)$ through convex lifting~\cite{candes2015cdp,gross2017improved,huangli2026optimal}.
Introducing the lifted variable \(\bm X=\bm x\bm x^*\) and the lifted ground truth \(\bm X_\star=\bm x_\star\bm x_\star^*\), the phase retrieval problem in \cref{eq: phase retrieval model} can be formulated as a convex program over the positive semidefinite (PSD) cone:
\begin{equation} \label{eq: PhaseLift problem}
    \text{find }\bm X\succeq\bm 0 \quad\text{subject to}\quad \mathcal A(\bm X)=\bm y,
\end{equation}
where $\mathcal A(\bm X) = \bigl(\langle\bm a_i\bm a_i^*,\bm X\rangle\bigr)_{i=1}^m$ is the lifted rank-one measurement operator.
In the presence of noise, PhaseLift-type variants may incorporate trace
information, replace the equality constraint with
$
\|\mathcal A(\bm X)-\bm y\|_2\leq \|\bm e\|_2,
$
or use a data-fitting loss.
However, no existing
stable recovery guarantee for convex CDP programs attains the noise
scaling achieved with i.i.d.\ Gaussian measurements; see
\cref{tab:stability-comparison}.
Moreover, to our knowledge, no global benign-landscape guarantee has been
established for the CDP model, either at factor width one or under overparameterization, despite encouraging empirical evidence from related models~\cite{sun2018geometric,li2021incremental}.
Closing these gaps has been a long-standing challenge in the phase retrieval literature.

\begin{table*}[t]
\centering
\caption{Comparison of stable recovery guarantees for PhaseLift-type convex programs.
Here, \(\widehat{\bm X}\) denotes a solution of the corresponding program, 
and
\(\bm y=\mathcal A(\bm X_\star)+\bm e\) is the noisy observation.}
\label{tab:stability-comparison}
\small
\setlength{\tabcolsep}{4pt}
\renewcommand{\arraystretch}{1.3}
\begin{tabularx}{\textwidth}{
    >{\raggedright\arraybackslash}p{2.3cm}
    >{\raggedright\arraybackslash}p{3.0cm}
    >{\raggedright\arraybackslash}p{2.4cm}
    >{\raggedright\arraybackslash}X}
\toprule
Reference & Setting & Complexity & Error bound \\
\midrule
\addlinespace
\cite{candes2014quadratic}
& i.i.d.\ Gaussian
& $m=O(n)$
& $\|\widehat{\bm X}-\bm X_\star\|_F
   \lesssim\|\bm e\|_2/\sqrt m$ \\
\addlinespace
\cite{krahmer2020complex}
& i.i.d.\ complex \(K\)-sub-Gaussian
&
\(m=O(C(K,\beta)n)\)
&
If
\(\lvert\mathbb E d_1^2\rvert^2\leq1-\beta\) and
[(i) \(\|\bm x_\star\|_\infty\leq\|\bm x_\star\|_2/81\) or (ii) $\mathbb{E}|d_1|^4\geq1+\beta$], then
$
\|\widehat{\bm X}-\bm X_\star\|_F
\lesssim
c(K,\beta)
\frac{\|\bm e\|_2}{\sqrt m}.
$ \\
\midrule
\cite{lili2021fourier}
& CDP
& $L=O(\log^4 n)$
& $\|\widehat{\bm X}-\bm X_\star\|_F
   \lesssim \log^{3/2} n\,\|\bm e\|_2$ \\
\addlinespace
\cite{huangwen2026phaselift}
& CDP
& $L=O(\log^2 n)$
& Upper: $\|\widehat{\bm X}-\bm X_\star\|_F
  \lesssim\sqrt{\|\bm e\|_2\log n/\sqrt m}$ \\[2pt]
&
& $L=O(\log n)$
& Lower: $\exists \bm{e}$ s.t. $\|\widehat{\bm X}-\bm X_\star\|_F
  \gtrsim\|\bm e\|_2/\sqrt{m\log n}$ \\
\midrule
\textbf{This work}
& CDP
& $L=O(\log n)$
& $\|\widehat{\bm X}-\bm X_\star\|_F
   \lesssim \log^2(2n)\,\|\bm e\|_2/\sqrt m$ \\
\bottomrule
\end{tabularx}
\end{table*}

A key missing ingredient is stability of the lifted measurement operator,
whose precise technical expression is the following uniform lower isometry:
\begin{equation} \label{eq: Gaussian-stability}
\|\bm{X}-\bm{X}_\star\|_F \lesssim \frac{\|\mathcal{A}(\bm{X} - \bm{X}_\star)\|_2}{\sqrt{m}}, \quad \forall \bm{X} \succeq 0.
\end{equation}
For Gaussian measurements, \cref{eq: Gaussian-stability} holds with \(m=O(n)\)~\cite{candes2014quadratic}, showing that \(\mathcal A\) is
stably injective on the PSD secant set
$
\{\bm X-\bm X_\star:\bm X\succeq\bm 0\}.
$
A CDP analogue has two consequences. 
First,
it yields the conjectured linear-noise and \(m^{-1/2}\) stable-recovery
scaling~\cite{soltanolkotabi2014algorithms}, up to logarithmic factors.
Second, together with a suitable
spectral upper bound, it enables the deterministic framework
of~\cite{mcrae2026benign} to establish a benign landscape for the
overparameterized nonconvex loss.
The remaining deductions are deterministic
(see \cref{sec:deterministic-landscape}), and the main technical challenge is to
establish \cref{eq: Gaussian-stability} for the CDP model with high probability at low mask complexity.

For Gaussian measurements, \cref{eq: Gaussian-stability} follows from $\ell_1$-norm bounds on both tangent and normal components of $\bm{X} - \bm X_\star$~\cite{candes2014quadratic}.
For CDPs, within-pattern dependence precludes probabilistic tools based on row independence, which has led existing analyses to use an $\ell_2$ tangent bound~\cite{candes2015cdp,soltanolkotabi2014algorithms}. 
This mismatch with the $\ell_1$ normal bound incurs a dimension-dependent loss, leading to a suboptimal scaling. 
Our key new ingredient is a uniform operator-norm bound over row subsets of the
CDP measurement matrix (see \cref{prob:prop-fixed-row-subset}), proved using a bound for suprema of chaos
processes from~\cite{krahmer2014suprema}.
It allows us to erase a controlled set of data-dependent rows carrying large
normal-component measurements while preserving tangent injectivity, and yields
a matching \(\ell_2\) normal bound outside the erased rows.
\Cref{sec:proof-overview} provides the
technical overview.

The main contribution of this paper is to establish the stability result
for the CDP model, with consequences for both convex and nonconvex recovery.
The following main results hold with high probability at the optimal mask complexity $L \asymp \log n$:
\begin{itemize}
    \item \emph{Uniform lower isometry:} \begin{equation*}
        \frac{1}{m}
        \left\|\mathcal{A}(\bm{X}-\bm{X}_\star)\right\|_2^2
        \gtrsim
        \frac{1}{\log^4(2n)}
        \|\bm{X}-\bm{X}_\star\|_F^2,
    \end{equation*}
    simultaneously for every $\bm{X}\succeq\bm{0}$; see
    \cref{thm:stability}.

    \item \emph{Stable convex recovery:} \begin{equation*}
        \|\widehat{\bm X}-\bm X_\star\|_F
        \lesssim \log^2(2n)\frac{\|\bm{e}\|_2}{\sqrt m},
    \end{equation*}
    for every PSD point satisfying $\|\mathcal A(\widehat{\bm X})-\bm y\|_2\leq \|\bm{e}\|_2$; see \cref{cor:noisy-recovery}.

   \item \emph{Benign landscape:}
In the noiseless case, the landscape of the factorized least-squares loss is benign under mild
overparameterization \(r=O(\log^5(2n))\); see
\cref{cor:benign-landscape-noisy}.
\end{itemize}

The paper is organized as follows.  \Cref{sec:preliminaries} specifies the CDP
model and fixes notation.  \Cref{sec:mainresults} states the uniform 
lower isometry theorem and its consequences for stable convex recovery and
overparameterized optimization.  \Cref{sec:mainproof} proves the lower
isometry, with the key new input being a uniform row-subset bound.
\Cref{sec:probabilistic} develops the chaos-process argument that establishes
this bound.  The remaining ingredients adapt existing techniques:
\cref{sec:dual_cert_proof} modifies a golfing construction to produce the
required pure-range dual certificate, while \cref{sec:deterministic-landscape}
applies a deterministic landscape framework to obtain a noisy critical-point
error bound and a benign noiseless landscape.
\Cref{sec:conclusion} summarizes the results and discusses open questions.

\section{Preliminaries}
\label{sec:preliminaries}

\subsection{Coded Diffraction Model}
\label{sec:model}
We now specify the random coded diffraction model. Let $\bm x_\star\in \mathbb{C}^n$ be the unknown signal, and let $\bm D_1,\ldots,\bm D_L$ be $L$ independent masks of the form $\bm D_\ell = \operatorname{diag}(d_{\ell,0},\ldots,d_{\ell,n-1})$, where all entries are independent copies of a random variable \(d\).
Let \(\bm F=[\bm f_0,\ldots,\bm f_{n-1}]\in\mathbb C^{n\times n}\)
denote the unnormalized Fourier matrix, so that
\(\bm F^*\bm F=n\bm I\).
The noisy CDP measurements are given by
\begin{equation}
y_{\ell k}
=
\left|
\left\langle \bm f_k,\bm D_\ell\bm x_\star\right\rangle
\right|^2
+e_{\ell k},
\qquad
1\leq \ell\leq L,\quad 0\leq k<n.
\end{equation}
Each mask produces one Fourier-intensity diffraction pattern, giving $m=nL$ measurements in total.

Introducing the lifted signal $\bm X_\star=\bm x_\star\bm x_\star^*$, we define $\bm a_{\ell k}=\bm D_\ell^*\bm f_k$ and the linear measurement operator $\mathcal A:\mathbb H_n\to\mathbb R^{nL}$ by
\begin{equation}
[\mathcal A(\bm H)]_{\ell k} = \operatorname{tr}(\bm a_{\ell k}\bm a_{\ell k}^*\bm H), \qquad 1\leq\ell\leq L,\quad 0\leq k<n. \label{eq:cdp-operator}
\end{equation}
Thus, the observation model can be written compactly as
$
\bm y=\mathcal A(\bm X_\star)+\bm e.
$

Throughout the paper, we impose the following random-mask assumption, which is commonly used in theoretical analyses of coded diffraction measurements~\cite{candes2015cdp,gross2017improved}.
\begin{assumption}[Bounded CDP masks]
\label{ass:mask}
The mask variable \(d\) satisfies
\begin{equation}
  d\stackrel{\mathrm d}{=}-d,\qquad
  \mathbb E d=\mathbb E d^2=0,\qquad
  \mathbb E|d|^2=1,\qquad
  \mathbb E|d|^4=2,\qquad
  |d|\leq M\quad\text{a.s.}
  \label{eq:mask-assumption}
\end{equation}
\end{assumption}

One concrete choice is the bounded erasure mask
\begin{equation}
  \mathbb P\{d=0\}=\frac{1}{2},\qquad
  \mathbb P\{d=\sqrt{2}\}
  =\mathbb P\{d=-\sqrt{2}\}
  =\mathbb P\{d=i\sqrt{2}\}
  =\mathbb P\{d=-i\sqrt{2}\}=\frac{1}{8}.
  \label{eq:mask-example}
\end{equation}
The normalization $\mathbb E|d|^2=1$ is inessential, and the results can be rescaled to accommodate $\mathbb E|d|^2=\nu$ for any $\nu>0$.

\subsection{Notation}
Throughout, we fix a unit vector \(\bm x_\star\), independent of the
masks, and write \(\bm X_\star=\bm x_\star\bm x_\star^*\). The tangent
space at \(\bm X_\star\) to the manifold of rank-one Hermitian matrices is
\[
  \tang=\{\bm x_\star\bm h^*+\bm h\bm x_\star^*:\bm h\in\mathbb C^n\}.
\]
We denote the orthogonal projections onto \(\tang\) and its orthogonal
complement by \(\proj{\tang}\) and \(\proj{\tangp}\), respectively.
For any \(\bm X\succeq0\), let
\begin{align*}
  \bm H&=\bm X-\bm X_\star=\HT+\HTp,\qquad
  \HT=\proj{\tang}\bm H,\\
  \HTp&=\proj{\tangp}\bm X
  =(\bm I-\bm X_\star)\bm X(\bm I-\bm X_\star)\succeq0.
\end{align*}

Unless stated otherwise, \(C,c\), and their decorated variants denote positive
constants that may change from line to line. 
For nonnegative
quantities \(a\) and \(b\), we write \(a\lesssim b\), or 
\(a=O(b)\), if \(a\leq Cb\); similarly, \(a\gtrsim b\), or
\(a=\Omega(b)\), if \(a\geq cb\), and \(a\asymp b\) if both bounds hold.
Throughout, \(\log\) denotes the natural logarithm.
We use \(\widetilde O\), \(\widetilde\Omega\), and \(\widetilde\Theta\) for the
corresponding relations up to multiplicative factors that are polynomial in
the logarithms of $n$ and $m$.
Subscripts on constants and
comparison symbols record or emphasize the indicated parameter dependence
without asserting that it is exhaustive; for example, \(C_M\) and
\(\lesssim_M\) highlight dependence on \(M\).

For a vector \(\bm v\), \(\|\bm v\|_p\), \(p\in\{1,2,\infty\}\), denotes
the usual \(\ell_p\)-norm, while \(\|\bm v\|_0\) denotes its number of nonzero
entries. For a matrix \(\bm H\), \(\|\bm H\|_F\) and
\(\|\bm H\|_{\mathrm{op}}\) denote the Frobenius and operator norms,
respectively; the latter notation also applies to linear maps. We use
\(|\cdot|\) for scalar absolute values and \(|E|\) for the cardinality of a
finite set \(E\).
For an integer $N$, let $[N]:=\{1 \dots N\}$.

\section{Main Results}
\label{sec:mainresults}

\subsection{Uniform Lower Isometry}

\label{sec:measurement-stability}
We first establish stability of the lifted CDP measurement operator
\(\mathcal A\) defined in \cref{eq:cdp-operator}. Its precise technical
expression is a uniform lower isometry: with high probability,
\(L \asymp \log n\) random masks yield an operator satisfying a bound of the
form \cref{eq: Gaussian-stability} over the PSD secant set. In other words, over
the entire PSD cone, the measurement discrepancy
\(\|\mathcal A(\bm X-\bm X_\star)\|_2\) uniformly controls the Frobenius error
\(\|\bm X-\bm X_\star\|_F\).

\begin{theorem}
  [Uniform lower isometry over the PSD secant set]
  \label{thm:stability}
Fix \(\omega\geq1\). There are constants \(0<C_1<C_2\) and \(C,c>0\),
depending only on \(M\) and \(\omega\), such that, if
\begin{equation}
  n\geq C\log^4(2n),
  \qquad C_1\log n\leq L\leq C_2\log n,
  \label{cert:eq:stability-assumptions}
\end{equation}
then, with probability at least
\begin{equation}
  1-Cn^{-\omega},
  \label{cert:eq:stability-probability}
\end{equation}
one has, simultaneously for every $\bm{X} \succeq 0$,
\begin{equation}
 \frac1m\|\mathcal A(\bm{X} -\bm{X}_\star)\|_2^2
 \geq\frac{c}{\log^4(2n)}\|\bm{X}-\bm{X}_\star\|_F^2.
 \label{eq:measurement-stability}
\end{equation}
\end{theorem}
For comparison, the standard uniform isometry in low-rank matrix sensing is
the \emph{restricted isometry property} (RIP). A linear operator \(\mathcal A\) satisfies the rank-\(r\) RIP if
\[
  (1-\delta_r)\|\bm H\|_F^2
  \leq \frac{1}{m}\|\mathcal A(\bm H)\|_2^2
  \leq (1+\delta_r)\|\bm H\|_F^2,
\]
for every matrix \(\bm H\) with \(\operatorname{rank}(\bm H)\leq r\); see,
e.g.,~\cite{recht2010guaranteed,candes2011tight}. 
For Gaussian matrix sensing, the rank-\(r\) RIP holds at \(m=O(nr)\)~\cite[Theorem~2.3]{candes2011tight}.
In particular, an RIP
over all full-rank matrices holds at \(m=O(n^2)\).
By contrast, the ground truth $\bm{X}_\star$ is fixed before sampling, but
\cref{thm:stability} holds for every PSD candidate, and the resulting secant set may contain full-rank matrices.
Thus, a uniform lower isometry with \(L\asymp\log n\) masks and
\(m\asymp n\log n\) measurements may appear surprising. There is no
contradiction with RIP: our result is a one-sided bound on the structured set
$
  \{\bm X-\bm X_\star:\bm X\succeq\bm 0\}.
$
The positive
semidefiniteness of \(\bm X\), \(\bm X_\star\), and the measurement matrices
\(\bm a_i\bm a_i^*\) restricts this secant geometry and makes the lower bound
possible; see \cref{sec:proof-overview} for the proof.

\subsection{Stable Recovery of Convex Programs}

The uniform lower isometry \cref{eq:measurement-stability} in \cref{thm:stability} holds for every
PSD estimator $\widehat{\bm X}\succeq0$, regardless of rank; thus, trace information that helps promote low-rank solutions may not be required.
Instead, the key to recovering $\bm X_\star$ is to find a PSD $\widehat{\bm X}$ that fits the noiseless measurements $\mathcal A(\bm X_\star)$ well, which can be achieved by fitting the noisy measurements $\bm y$.
\begin{corollary}[Stable Recovery for Convex Programs]
\label{cor:noisy-recovery}
Suppose that the uniform lower isometry bound \eqref{eq:measurement-stability} from \cref{thm:stability} holds.
Let
$\bm y=\mathcal A(\bm X_\star)+\bm e$ with $\|\bm e\|_2\leq\delta$.
Then there exists a constant $C>0$ such that every $\widehat{\bm X}\succeq0$ satisfying
\begin{equation} \label{eq: noise-equality}
  \|\mathcal{A}(\widehat{\bm X})-\boldsymbol{y}\|_2 \leq \delta,
\end{equation}
obeys
\begin{equation}
  \|\widehat{\bm X}-\bm X_\star\|_F
  \leq C\log^2(2n)\frac{\delta}{\sqrt{nL}}.
  \label{eq:noisy-recovery}
\end{equation}
\end{corollary}

\begin{proof}
The triangle inequality and feasibility give
\[
  \|\mathcal A(\widehat{\bm X}-\bm X_\star)\|_2
  \leq \|\mathcal A(\widehat{\bm X})-\bm y\|_2+\|\bm e\|_2
  \leq 2\delta.
\]
Applying \eqref{eq:measurement-stability} with
$\bm X=\widehat{\bm X}$ and using $m=nL$ proves
\eqref{eq:noisy-recovery}.
\end{proof}
It remains to verify that the condition in \cref{eq: noise-equality} holds for specific PSD convex programs.
For example, the condition in \cref{eq: noise-equality} can be ensured by either a feasibility constraint
\begin{equation} \label{eq:noisy-psd-feasibility}
 \text{find} \quad \bm X\succeq 0 \quad \text{subject to} \quad \|\mathcal A(\bm X)-\bm y\|_2\leq \delta,
\end{equation}
or a least-squares loss:
\begin{equation}\label{eq:noisy-psd-least-squares} \min_{\bm X\succeq 0}\|\mathcal A(\bm X)-\bm y\|_2^2, \end{equation}
where the bound follows from the feasibility of $\bm X_\star$ and the least-squares optimality of $\widehat{\bm X}$.
The condition in \cref{eq: noise-equality} also holds for the convex program used in the current state-of-the-art stable recovery result for CDPs~\cite{huangwen2026phaselift}:
\begin{equation}\label{eq:phaselift-trace} \min_{\bm X\succeq \bm 0}\operatorname{tr}(\bm X) \quad \text{subject to} \quad \|\mathcal A(\bm X)-\bm y\|_2\leq\delta.
 \end{equation}
Although their approach may extend to general models without
a uniform lower isometry such as
\cref{eq: Gaussian-stability,eq:measurement-stability}, the resulting error bound scales as
\(\sqrt{\delta\log n/\sqrt{nL}}\), exhibiting square-root rather than linear
dependence on the noise level. In the low-noise regime, this is weaker than
the bound \(\log^2(2n)\delta/\sqrt{nL}\) established in
\cref{eq:noisy-recovery}.

We note that Jaganathan et al.~\cite{jaganathan2015masks} obtained an
error bound linear in \(\|\bm e\|_\infty\) using deterministic masks, but under additional
assumptions on \(\bm x_\star\).

\subsection{Application to Overparameterized Optimization}
\label{sec:overparameterized-result}

Instead of solving convex programs such as
\cref{eq:noisy-psd-least-squares}, which optimize over \(n\times n\)
matrices, we may consider the factorized least-squares objective
\begin{equation}
F_r(\bm V)
=
\frac{1}{4nL}
\|\mathcal A(\bm V\bm V^*)-\bm y\|_2^2,
\qquad
\bm V\in\mathbb C^{n\times r},
\label{eq:factorized-objective}
\end{equation}
which has lower storage and per-iteration computational costs when \(r\ll n\).
The complex factor space is viewed as a real Euclidean space in optimization; see
\cref{eq:factorized-objective-real} in
\cref{sec:deterministic-landscape}.
Since the ground truth \(\bm X_\star\) has rank one, \(r=1\) is sufficient
to represent \(\bm X_\star\) at a global minimizer in the noiseless case. The
resulting objective is the standard quadratic intensity loss used in phase
retrieval~\cite{candes2015wirtinger,chen2017solving}.
However,
the factorized objective is nonconvex and may contain spurious local minima
that may trap local-search methods.
Intuitively, increasing the factor width $r$ enlarges the search space and provides additional descent directions, which may help eliminate spurious local minima.
This technique is called \emph{overparameterization}, with
applications in low-rank recovery, group synchronization, and neural-network
training~\cite{zhang2021overparameterized,ma2023overparametrization,
mcrae2024synchronization,nguyen2017loss,du2018power}.

For phase retrieval, Li et al. demonstrated the empirical benefits of incremental rank search with their \texttt{IncrePR} algorithm~\cite{li2021incremental}.
More recently, McRae~\cite{mcrae2026benign,mcrae2026amplitude} analyzed how increasing the factor
width improves the landscapes of phase-retrieval objectives under
deterministic conditions, which include a local spectral upper bound and
a global lower isometry.
For the CDP model, the spectral upper bound follows directly from a
standard coherence estimate, while the uniform lower isometry in
\cref{thm:stability} provides the remaining ingredient.
Applying McRae's
framework gives the following guarantee under noise; the noiseless case is
included by setting \(\bm e=\bm 0\), for which
\(r=O(\log^5(2n))\) is sufficient for a benign landscape.

\begin{corollary}[CDP landscape under noise]
    \label{cor:benign-landscape-noisy}
    Under the assumptions of \cref{thm:stability}, let
    \(\bm y=\mathcal A(\bm X_\star)+\bm e\), and let \(c\) be the constant in
    \eqref{eq:measurement-stability}. Define
    \begin{equation}
        \tau:=\frac{2\log^4(2n)}{mc}
        \|\mathcal A^*(\bm y)\|_{\mathrm{op}}-2.
    \end{equation}
    With probability at least \eqref{cert:eq:stability-probability}, for every
    \(r>\max\{\tau,0\}\), each second-order critical point \(\bm V\) of
    \(F_r\) satisfies
    \begin{equation}
        \label{eq:noisy-nonconvex-errorbound}
        \|\bm V\bm V^* - \bm X_\star\|_F
        \leq \frac{r+2}{r-\tau}
        \frac{\log^2(2n)}{\sqrt{mc}}\,\|\bm e\|_2.
    \end{equation}
    In particular, on the same event, if \(\bm e=\bm 0\), then there exists a constant $C_{M,\omega}>0$ such that every \(r\)
satisfying
\begin{equation}
r\geq \frac{2C_{M,\omega}}{c}\log^5(2n)
\label{eq:landscape-width}
\end{equation}
has the following property: every second-order critical point satisfies
\(\bm V\bm V^*=\bm X_\star\), and every nonglobal critical point has a
direction of strictly negative curvature.
\end{corollary}

\Cref{eq:landscape-width} shows that mild overparameterization of order \(r = O(\log^5(2n))\) is sufficient to guarantee a benign landscape for \cref{eq:factorized-objective} in the noiseless case.
This polylogarithmic factor width is enabled by the uniform lower isometry over
the PSD secant set in \cref{thm:stability}; the previous stable recovery guarantees do not directly provide the
uniform input required by this landscape framework.
See
\cref{sec:deterministic-landscape} for a detailed discussion and the proof of
\cref{cor:benign-landscape-noisy}.

\section{Proof of the Uniform Lower Isometry Theorem}
\label{sec:mainproof}
This section proves the uniform lower isometry result in \cref{thm:stability}.
Throughout, set \(m=nL\) and fix a failure probability
\(0<\eta<1/4\).  For brevity, we use 
$
\Gamma_{n,m}$ and 
$
\Lambda_{n,m,\eta}$
to denote the logarithmic factors (defined later in \eqref{eq:eq-log-factors}). 
Let \(\bm{\Phi}\in\mathbb C^{m\times n}\) be the measurement matrix with rows
\(\bm a_i^*\), indexed either by \(i\in[m]\) or by
\((\ell,k)\in[L]\times \{0, \dots, n-1\}\). 
 For \(E\subset[m]\), let \(\bm\Phi_E \in \mathbb{C}^{|E|\times n}\)
denote the submatrix formed by the rows indexed by \(E\), and let
$\mathcal{A}_E$ denote the lifted measurement operator restricted to the rows in \(E\):
$$
\mathcal{A}_E(\bm H) := \bigl(\langle\bm a_i\bm a_i^*,\bm H\rangle\bigr)_{i\in E} \in \mathbb{C}^{|E|}.
$$  
Then
\begin{equation}
  (\bm{\Phi}\bm{h})_{\ell k}=\bm{a}_{\ell k}^*\bm{h},
  \qquad
  \bm{A}_i=\bm{a}_i\bm{a}_i^*,
  \qquad
  \bm{\Phi}_E^*\bm{\Phi}_E=\sum_{i\in E}\bm{A}_i.
  \label{prob:eq-analysis-matrix}
\end{equation}

\subsection{Proof Overview}
\label{sec:proof-overview}

Fix \(\bm X\succeq\bm0\), let \(\bm H=\bm X-\bm X_\star\), and decompose
the error into tangent and normal components
\begin{equation*}
  \bm H=\HT+\HTp,
  \qquad
  \HT=\proj{\tang}\bm H,
  \qquad
  \HTp=\proj{\tangp}\bm X\succeq\bm0.
\end{equation*}
Existing CDP results give the \(\ell_2\) tangent lower bound 
$\|\mathcal{A}(\HT)\|_2 / \sqrt{m} \gtrsim \|\HT\|_F$
in
\cref{prop:tangent-lower}; see~\cite[Prop.~2]{huangli2026optimal}.
For the PSD normal component, the measurements are instead controlled in the $\ell_1$ norm:
\(m^{-1}\|\mathcal A(\HTp)\|_1\asymp\operatorname{tr}(\HTp)\); see~\cite[Lem.~3.3]{candes2015cdp}.  Directly converting the \(\ell_1\) bound to an
\(\ell_2\) upper bound would give
\begin{equation*}
  \frac{\|\mathcal A(\HTp)\|_2}{\sqrt m}
  \leq\frac{\|\mathcal A(\HTp)\|_1}{\sqrt m}
  \lesssim\sqrt m\,\operatorname{tr}(\HTp).
\end{equation*}
This would introduce a dimension-dependent loss $\sqrt m$.
A dimension-free \(\ell_2\) upper bound cannot simply be imposed uniformly
on the normal PSD cone: the familiar rank-one obstruction already prevents
such an \(\ell_2\)-RIP~\cite[Appendix, pp.~1271--1272]{candes2013phaselift}.

The key new ingredient is the uniform row-subset estimate
\begin{equation}
\sup_{\substack{E\subset[m]\\ |E|\leq s}}
\|\bm\Phi_E\|_{\mathrm{op}}^2
\lesssim
n+\sqrt{ns\Lambda_{n,m,\eta}}
+s\Lambda_{n,m,\eta},
\label{prob:eq-overview-row-subset}
\end{equation}
proved in \cref{sec:probabilistic}.  
For a positive semidefinite normal component \(\HTp\), let
\(E_{\HTp}\subset[m]\) denote the indices of its \(s\) largest measurements,
and let \(E_{\HTp}^c\) denote the complementary index set.
Since
\eqref{prob:eq-overview-row-subset} is uniform over all row subsets, it also
applies to \(E_{\HTp}\).
Taking \(s\asymp n/\Lambda_{n,m,\eta}\), the estimate shows that 
 $\|\mathcal{A}_{E_{\HTp}^c}(\HTp)\|_\infty$  is bounded in terms of
\(\Lambda_{n,m,\eta}\operatorname{tr}(\HTp)\); after erasing, we have
$$
\|\mathcal{A}_{E_{\HTp}^c}(\HTp)\|_2^2 \leq \|\mathcal{A}_{E_{\HTp}^c}(\HTp)\|_1\|\mathcal{A}_{E_{\HTp}^c}(\HTp)\|_\infty
\lesssim m \Lambda_{n,m,\eta} \operatorname{tr}^2(\HTp),
$$ 
leading to a \(\ell_2\) estimate for the normal component.
Together with the bound $\max_{1\leq i \leq m} |\bm{a}_i^* \bm{x}_\star|^2 \lesssim \log(m)$ in \cref{prob:eq-coherence},
\cref{prob:eq-overview-row-subset} also shows that erasing \(E_{\HTp}\) deletes only a controlled fraction
of the tangent measurement energy.  Consequently, tangent injectivity
survives on \(E_{\HTp}^c\), and we establish the matching \(\ell_2\) control of both the tangent component and the normal component.

The following argument follows the standard tangent--normal decomposition and
dual-certificate strategy from PhaseLift-type analyses; see, e.g.,~\cite{candes2015cdp,mcrae2026benign}.
Restricted to $E_{\HTp}^c$, tangent injectivity gives
$\|\HT\|_F \lesssim \|\mathcal{A}_{E_{\HTp}^c}(\HT)\|_2 / \sqrt{m}$,
and the triangle inequality gives
$$
\|\HT\|_F \lesssim \frac{\|\mathcal{A}_{E_{\HTp}^c}(\bm{H})\|_2}{\sqrt{m}} + \frac{\|\mathcal{A}_{E_{\HTp}^c}(\HTp)\|_2}{\sqrt{m}} \lesssim \frac{\|\mathcal{A}(\bm{H})\|_2}{\sqrt{m}} + \Lambda_{n,m,\eta}^{\frac{1}{2}} \operatorname{tr}(\HTp).
$$
The
certificate in \cref{prop:pure-range-certificate}
provides the reverse bridge:
\[
\|\HTp\|_F \leq  \operatorname{tr}(\HTp)
\lesssim
\frac{\|\mathcal A(\bm H)\|_2}{\sqrt m}
+\varepsilon\|\HT\|_F.
\]
Choosing \(\varepsilon \lesssim \Lambda_{n,m,\eta}^{-\frac{1}{2}}\) and closing these two inequalities
controls both \(\|\HT\|_F\) and \(\|\HTp\|_F\) by $\|\mathcal{A}(\bm{H})\|_2$, thereby establishing the
lower isometry in \cref{thm:stability}. 

The rest of the proof is organized as follows.  
We first introduce the uniform row-subset bound in
\cref{prob:prop-fixed-row-subset}, then derive the PSD measurement bound in
\cref{prob:lem-psd-size}, the erasure-robust tangent \(\ell_2\)
bound
in \cref{thm:tangent-lower-outside-large}, and the PSD \(\ell_2\) bound outside
the erased rows in \cref{lem:psd-l2-after-row-removal}.  We then use these results to prove \cref{thm:stability}.

\subsection{Key Ingredient: A Uniform Row-Subset Bound}

The uniform row-subset bound quantifies how much measurement energy can
concentrate on a small set of rows.  Related questions arise in the study of
robustness under arbitrary measurement erasures~\cite{laska2011democracy,han2017robustness}, and also in quantized inverse
problems~\cite{dirksen2021nonGaussian,chenYuanOneBitPhaseRetrieval}.
If a $m \times n$ matrix \(\bm A \) has independent isotropic sub-Gaussian rows, then, with high
probability,
\begin{equation}
    \sup_{\substack{E\subset[m]\\ |E|\le s}}
    \|\bm A_E\|_{\mathrm{op}}^2
    \lesssim n+s\log(em/s).
    \label{eq:subgaussian-row-subset}
\end{equation}
See, e.g.,~\cite[Lemma 32]{chenYuanOneBitPhaseRetrieval} for its use in one-bit phase retrieval.
Our
estimate provides an analogous row-subset bound for the structured and
dependent measurements arising from coded diffraction patterns.

\begin{proposition}[Uniform row-subset bound up to size \(s\)]
\label{prob:prop-fixed-row-subset}
There is a constant \(C_M>0\), depending only on the mask bound \(M\), such
that the following holds.  Fix an integer \(1\le s\le n\) and $\eta \in(0,1/4)$.  
With probability at least \(1-\eta\),
\begin{equation}
  \|\bm{\Phi}_E\|_{\mathrm{op}}^2
  \le n+C_M\left(\sqrt{ns(\Lameta)}+s(\Lameta)\right)
  \label{prob:eq-fixed-row-subset}
\end{equation}
simultaneously for every \(E\subset[m]\) with \(|E| \leq s\).

\end{proposition}

The restriction $s \leq n$ can be generalized to $s \leq m$ by introducing logarithmic factors of $s$, but the simpler form is sufficient for our purposes.
We also introduce the following shorthand notation for logarithmic factors:
\begin{equation} \label{eq:eq-log-factors}
\Gamma_{n,m}=\log^2(2n)\log^2(2m),\qquad 
\Lambda_{n,m,\eta} = \Gamma_{n,m} + \log\frac{2}{\eta},
\end{equation}
where $\Lambda_{n,m,\eta} \geq 1$ because $\eta \in (0,1/4)$.
The dependent row structure makes the proof of \cref{prob:prop-fixed-row-subset} more involved than that of \cref{eq:subgaussian-row-subset}; we defer it to \cref{sec:probabilistic}.
In the remainder of this section, we establish the uniform lower isometry theorem from the uniform row-subset bound.

\subsection{PSD Measurements Outside the Erased Rows}

For \(\bm P\succeq\bm 0\), choose a permutation \(\pi_{\bm P}\) of \([m]\)
such that
\begin{equation}
  \langle\bm A_{\pi_{\bm P}(1)},\bm P\rangle
  \leq\cdots\leq
  \langle\bm A_{\pi_{\bm P}(m)},\bm P\rangle.
  \label{prob:eq-psd-measurement-order}
\end{equation}
If the ordering is not unique, choose any such ordering.
For a fixed $s \in [1,m]$, define
\begin{equation}
  E_{\bm P}:=\{\pi_{\bm P}(m-s+1),\ldots,\pi_{\bm P}(m)\},
  \qquad
  \zeta_{\bm P}:=
  \begin{cases}
    \displaystyle
    \frac{\langle\bm A_{\pi_{\bm P}(m-s+1)},\bm P\rangle}
         {\operatorname{tr}(\bm P)},
      & \operatorname{tr}(\bm P)>0,\\[1ex]
    0, & \bm P=\bm 0.
  \end{cases}
  \label{eq:defofzetaandE}.
\end{equation}
For simplicity, we omit the subscript $s$ in the notation of $E_{\bm P}$ and $\zeta_{\bm P}$, though they depend on the choice of $s$.
\(E_{\bm P}\) collects the \(s\) largest values of
\(\langle\bm A_i,\bm P\rangle\), while
\(\zeta_{\bm P}\operatorname{tr}(\bm P)\) is the smallest value among them,
namely the \(s\)-th largest measurement.  Thus
\begin{equation}
  \langle\bm A_j,\bm P\rangle
  \leq\zeta_{\bm P}\operatorname{tr}(\bm P)
  \quad\text{for every }j\notin E_{\bm P} \Rightarrow \|\mathcal A_{E_{\bm P}^c}(\bm P)\|_\infty \leq \zeta_{\bm P}\operatorname{tr}(\bm P).
  \label{prob:eq-psd-selection-cutoff}
\end{equation}

\begin{lemma}[PSD measurement bound outside the erased rows]
\label{prob:lem-psd-size}
Suppose that \cref{prob:eq-fixed-row-subset} holds.
Then simultaneously for every \(\bm P\succeq\bm{0}\),
\begin{equation}
  \zeta_{\bm P}s
  \leq n+C_M\left(\sqrt{ns\Lambda_{n,m,\eta}}+s\Lambda_{n,m,\eta}\right).
  \label{prob:eq-selected-zeta-bound}
\end{equation}
The matrix \(\bm{P}\) may have arbitrary rank and may depend on all masks.
\end{lemma}

\begin{proof}
The assertion is immediate for \(\bm P=\bm 0\).  Otherwise,
\(\operatorname{tr}(\bm P)>0\), and
\begin{align}
  \zeta_{\bm P}\operatorname{tr}(\bm P)s
  &\leq \sum_{i\in E_{\bm P}}\langle\bm A_i,\bm P\rangle
   =\operatorname{tr}(\bm P\bm\Phi_{E_{\bm P}}^*\bm\Phi_{E_{\bm P}})
  \notag\\
  &\leq\operatorname{tr}(\bm P)
       \|\bm\Phi_{E_{\bm P}}\|_{\mathrm{op}}^2
  \leq\operatorname{tr}(\bm P)
       \left(n+C_M\sqrt{ns\Lambda_{n,m,\eta}} +C_M s\Lambda_{n,m,\eta}\right).
  \label{prob:eq-psd-cardinality-contradiction}
\end{align}
Dividing by \(\operatorname{tr}(\bm P)\) proves the result.
\end{proof}

\subsection{Tangent Energy on the Erased Rows and Erasure-Robust Tangent Injectivity}
We next show that erasing a small set of rows preserves tangent injectivity. The argument starts from a fixed-target
coherence estimate controlling the largest squared magnitude among the masked
Fourier coefficients.
\begin{lemma}[Fixed-target coherence]
\label{prob:lem-coherence}
Let \(\bm{x}_\star\in\mathbb C^n\) be a deterministic unit vector.
Then, with probability at least \(1-\eta\),
\begin{equation}
  \mu_\star:=\max_{1\le \ell\le L,0\le k\le n-1}|\bm{a}_{\ell k}^*\bm{x}_\star|^2
  \le\mustar.
  \label{prob:eq-coherence}
\end{equation}
\end{lemma}

\begin{proof}
For fixed \((\ell,k)\),
\(\bm{a}_{\ell k}^*\bm{x}_\star
=\bm{f}_k^*\bm{D}_\ell\bm{x}_\star\) is a sum of independent,
centered complex variables, and each summand is bounded by
\(M|\bm{x}_\star(j)|\).  Applying Hoeffding's inequality to the real and imaginary
parts gives
\begin{equation}
  \mathbb P\{|\bm{a}_{\ell k}^*\bm{x}_\star|\ge tM\}
  \le4e^{-t^2/4}.
  \label{prob:eq-complex-hoeffding}
\end{equation}
Taking \(t=2\sqrt{\log (4m/\eta)}\) and using a union bound over \(m\) rows gives the desired result.
\end{proof}

Combining this coherence estimate with the uniform row-subset bound in \cref{prob:prop-fixed-row-subset} gives a uniform upper bound on the tangent energy carried by any set of at most \(s\) rows.

Every \(\HT\in\tang\) can be represented as
\begin{equation}
  \HT=\bm{x}_\star\bm{h}^*+\bm{h}\bm{x}_\star^*,
  \qquad \bm{x}_\star^*\bm{h}\in\mathbb R,
  \qquad \|\bm{h}\|_2^2\le\frac12\|\HT\|_F^2.
  \label{prob:eq-tangent-representation}
\end{equation}
For each row index \(i\),
\begin{equation}
  \langle \bm{A}_i,\HT\rangle
  =2\operatorname{Re}\left(
      (\bm{a}_i^*\bm{x}_\star)\overline{(\bm{a}_i^*\bm{h})}
    \right).
  \label{prob:eq-tangent-row}
\end{equation}
On the event where \cref{prob:eq-fixed-row-subset} and
\eqref{prob:eq-coherence} both hold,
\eqref{prob:eq-tangent-representation} and \eqref{prob:eq-tangent-row} yield,
simultaneously for all \(E\subset[m]\) with \(|E|=s\) and all
\(\HT\in\tang\),
\begin{align}
  &\frac1m\sum_{i\in E}\langle \bm{A}_i,\HT\rangle^2
  \le\frac{4 \cdot \mustar}{m}\|\bm{\Phi}_E\bm{h}\|_2^2
  \nonumber\\
  &\le\frac{C_M \ueta}{L}
       \left(1+\sqrt{\frac{s\Lambda_{n,m,\eta}}{n}}
       +\frac{s\Lambda_{n,m,\eta}}{n}\right) 
  \cdot\|\HT\|_F^2.
  \label{prob:eq-fixed-cardinality-erasure}
\end{align}

To show that tangent injectivity survives these erasures, we compare the preceding energy bound with the following lower bound for the full measurements.

\begin{proposition}[Robust Injectivity~{\cite[Proposition~2, p.~9]{huangli2026optimal}}]
\label{prop:tangent-lower}
Under the bounded CDP mask assumptions, with probability at least $1-2n\exp\left(-c\frac{L}{M^8}\right)$, the estimate
\begin{equation}
  \frac{\|\mathcal A(\HT)\|_2^2}{m}
  \geq\frac14\|\HT\|_F^2.
  \label{eq:prelim-tangent-lower}
\end{equation}
holds simultaneously for every \(\HT\in\tang\), where  \(c>0\) is an absolute
constant.
\end{proposition}

Recall that \(\Lambda_{n,m,\eta}=\Lameta\).  Comparing
\eqref{prob:eq-fixed-cardinality-erasure} with \cref{eq:prelim-tangent-lower} shows that the tangent energy on the erased rows
can be absorbed by taking
\[
  L\gtrsim\ueta,
  \qquad s\asymp\frac{n}{\Lameta}.
\]
Indeed, these choices can make the right-hand side of
\eqref{prob:eq-fixed-cardinality-erasure} at most
\(\|\HT\|_F^2/8\).
The choice of \(s\) also yields the scale-free bound from \cref{prob:eq-selected-zeta-bound}:
$$
\sup_{\bm P\succeq\bm 0}\zeta_{\bm P}\lesssim\Lameta.
$$
The following theorem formalizes both conclusions.

\begin{proposition}[Erasure-robust tangent injectivity]
\label{thm:tangent-lower-outside-large}
Let \(0<\eta<1/4\), and set
\(s=\lfloor n/\Lambda_{n,m,\eta}\rfloor\).  There are
constants \(C_M,c_M>0\) such that, if
\begin{equation}
  n\geq 2\Lambda_{n,m,\eta},
  \qquad L\ge C_M\ueta,
  \label{prob:eq-a3-assumptions}
\end{equation}
then with probability at least
\begin{equation}
  1-2\eta-C_Mn e^{-c_ML},
  \label{prob:eq-a3-probability}
\end{equation}
the following bounds hold simultaneously for every \(\HT\in\tang\) and every
\(\bm P\succeq\bm 0\),
\begin{equation}
  \frac1m\sum_{i\notin E_{\bm{P}}}
       \langle \bm{A}_i,\HT\rangle^2
  \ge\frac18\|\HT\|_F^2,
  \label{prob:eq-a3-conclusion}
\end{equation}
\begin{equation}
  \max_{i\notin E_{\bm P}}\langle \bm A_i,\bm P\rangle
  \leq C_M\Lambda_{n,m,\eta} \operatorname{tr}(\bm P).
  \label{prob:eq-zeta}
\end{equation}
Here $s \in [2,n]$, and $E_{\bm P}$  is defined in \eqref{eq:defofzetaandE}.
\end{proposition}

\begin{proof}
First, \cref{prop:tangent-lower} holds with probability at least
\(1-C_Mn e^{-c_ML}\), where
\begin{equation}
  \frac{\|\mathcal A(\HT)\|_2^2}{m}
  \geq\frac14\|\HT\|_F^2
  \qquad\text{for every }\HT\in\tang.
\end{equation}
Second, the uniform row-subset bound \cref{prob:eq-fixed-row-subset} and the
coherence bound \cref{prob:eq-coherence} hold simultaneously with probability at
least \(1-2\eta\). 
Since \(n/\Lambda_{n,m,\eta}\geq2\), we have $s \in [2,n]$, and the elementary
inequality \(\lfloor x\rfloor\geq x/2\) for \(x\geq2\) yields
\begin{equation}
  \frac{n}{2\Lambda_{n,m,\eta}}
  \leq s\leq \frac{n}{\Lambda_{n,m,\eta}},
  \qquad
  \frac{n}{s}\leq2\Lambda_{n,m,\eta},
  \qquad
  \frac{s\Lambda_{n,m,\eta}}{n}\leq1.
  \label{prob:eq-rounded-sparsity}
\end{equation}
Dividing
\eqref{prob:eq-selected-zeta-bound} by \(s\) and using
\eqref{prob:eq-rounded-sparsity} gives
\begin{equation}
  \zeta_{\bm P}
  \leq \frac ns
       +C_M'\sqrt{\frac{n\Lambda_{n,m,\eta}}s}
       +C_M'\Lambda_{n,m,\eta}
  \leq C_M'\Lambda_{n,m,\eta},
  \label{prob:eq-rounded-zeta}
\end{equation}
where \(C_M' > 0\) is a constant depending only on \(M\).  
Furthermore, \eqref{prob:eq-rounded-sparsity} and
\eqref{prob:eq-fixed-cardinality-erasure} give
\begin{equation}
  \frac1m\sum_{i\in E_{\bm P}}\langle\bm A_i,\HT\rangle^2
  \leq \frac{3 C_M'\ueta}{L}\|\HT\|_F^2
  \leq \frac18\|\HT\|_F^2,
\end{equation}
after increasing the constant \(C_M\)  in
\eqref{prob:eq-a3-assumptions}.  
Subtracting
the erased-row energy from the full tangent lower bound proves
\eqref{prob:eq-a3-conclusion}; the upper bound in \eqref{prob:eq-zeta} follows
from \eqref{prob:eq-psd-selection-cutoff} and \eqref{prob:eq-rounded-zeta}.
\end{proof}

\subsection{PSD \(\ell_2\) Bound Outside the Erased Rows}
Combining the infinity-norm bound in \eqref{prob:eq-zeta} with the existing \(\ell_1\) bound on PSD matrices, we can now establish a uniform \(\ell_2\) bound outside the erased rows.
\begin{lemma}[PSD \(\ell_2\) bound outside the erased rows]
\label{lem:psd-l2-after-row-removal}
Suppose that the bound \eqref{prob:eq-zeta} from
\cref{thm:tangent-lower-outside-large} holds. Then there exists a constant
\(C_M>0\) such that, for every
\(\bm P\succeq\bm 0\),
\begin{equation}
  \frac1m\|\mathcal A_{E_{\bm P}^c}(\bm P)\|_2^2
  \leq C_M\Lambda_{n,m,\eta}\operatorname{tr}(\bm P)^2.
  \label{eq:good-psd-bound}
\end{equation}
\end{lemma}

\begin{proof}
For every \(i\notin E_{\bm P}\), nonnegativity and
\eqref{prob:eq-zeta} give
\[
  0\leq\langle\bm A_i,\bm P\rangle
  \leq C_M'\Lambda_{n,m,\eta}\operatorname{tr}(\bm P).
\]
Hence, using Fourier Parseval, \(m=nL\), and
\(\bm D_\ell^*\bm D_\ell\preceq M^2\bm I\),
\begin{align}
  \frac1m\|\mathcal A_{E_{\bm P}^c}(\bm P)\|_2^2
  &\leq C_M'\Lambda_{n,m,\eta}\operatorname{tr}(\bm P)
       \frac1m\sum_{i=1}^m\langle\bm A_i,\bm P\rangle \notag\\
  &=C_M'\Lambda_{n,m,\eta}\operatorname{tr}(\bm P)
       \frac1L\sum_{\ell=1}^L
       \operatorname{tr}(\bm D_\ell^*\bm D_\ell\bm P) \notag\\
  &
   \leq C_M\Lambda_{n,m,\eta}\operatorname{tr}(\bm P)^2.
\end{align}
\end{proof}

\subsection{Proof of the Main Theorem}
The preceding estimates provide a tangent lower bound and a PSD upper bound
outside the erased rows. To complete the lower isometry argument, we next
introduce a pure-range certificate that relates the PSD normal component to the
measurement residual and the tangent component. Its construction adapts the
improved golfing scheme of Huang and
Li~\cite[Definition~1 and Propositions~7]{huangli2026optimal}; see
\cref{sec:dual_cert_proof} for the proof.

\begin{proposition}[Pure-range certificate]
\label{prop:pure-range-certificate}
Suppose that the bounded CDP mask assumptions hold. Let \(n\geq2\), and fix a unit vector
\(\bm x_\star\in\mathbb C^n\). Fix \(\omega\geq1\) and \(0<\varepsilon<1\). 
There exists a constant $C=C(M,\omega)$, depending only on $M$ and $\omega$, such that if
$
L
 =
 \left\lceil
 C(M,\omega)
 \bigl(\log n+\log(1/\varepsilon)\bigr)
\right\rceil$,
then with probability at least \(1-n^{-\omega}\),
there is a vector \(\bm q\in\mathbb R^{nL}\) such that
\(\bm Z=(nL)^{-1/2}\mathcal A^*(\bm q)\) satisfies
\begin{equation}
  \|\bm q\|_2\lesssim_{M,\omega}
   1+\frac{\log(1/\varepsilon)}{\log n}
  ,\qquad
  \|\proj{\tang}\bm Z\|_F\leq\varepsilon,\qquad
  \proj{\tangp}\bm Z\preceq-\frac12(\bm I-\bm X_\star).
  \label{eq:prelim-pure-range-certificate}
\end{equation}
\end{proposition}
The upper bound \(L\lesssim\log n\) in \cref{thm:stability} is used mainly to match the normalization of the certificate and is not an intrinsic limitation of the CDP model.
A pure-range certificate with tangent--normal geometry similar to that in \cref{eq:prelim-pure-range-certificate} was already
constructed by
Cand\`es et al.~\cite[Lemma~3.6 and Section~3.4]{candes2015cdp} with a larger number of masks, \(L\gtrsim\log^4 n\).
\begin{proof}[Proof of \Cref{thm:stability}]
Set \(\eta=n^{-\omega}\). Since \(L\leq C_2\log n\), we have
\begin{equation}
  \Gamma_{n,m}\asymp\log^4(2n),
  \qquad
  \Lambda_{n,m,\eta}\leq C_\omega\Gamma_{n,m},
  \qquad
  \log\frac{4m}{\eta}\leq C_{\omega,C_2}\log n.
  \label{cert:eq:log-comparison}
\end{equation}
Thus, after choosing \(C\) and \(C_1\) sufficiently large, the conditions in
\eqref{cert:eq:stability-assumptions} imply the assumptions of the tangent
estimate and the certificate construction used below.

Fix \(\bm X\succeq\bm 0\), and write
\begin{equation*}
  \bm H=\bm X-\bm X_\star,
  \qquad \HT=\proj{\tang}\bm H,
  \qquad
  \HTp=(\bm I-\bm X_\star)\bm X
       (\bm I-\bm X_\star)\succeq\bm 0.
\end{equation*}
\paragraph{Dual certificate}
Set
\begin{equation} \label{eq:eps}
\varepsilon=c_0/\sqrt{\Gamma_{n,m}},
\end{equation}
where \(c_0>0\), independent of \(n\) and \(m\), will be chosen below.  Then
\(
  \log(1/\varepsilon)
  =\log(1/c_0)+\frac12\log\Gamma_{n,m}.
\)
Since \(L=O(\log n)\), this gives
$
  \log(1/\varepsilon) \lesssim \log n.
$ 
Let
\begin{equation}
  L_0=\left\lceil C(M,\omega)
  \bigl(\log n+\log(1/\varepsilon)\bigr)\right\rceil,
  \label{cert:eq:certificate-required-masks}
\end{equation}
be the number of masks used by \cref{prop:pure-range-certificate}. By choosing
\(C_1\) sufficiently large, \(L_0\leq L\); moreover,
\(L/L_0\leq C'\) because \(L\leq C_2\log n\) and \(L_0\gtrsim\log n\).
Apply \cref{prop:pure-range-certificate} to the first \(L_0\) masks, and let
\(\widetilde{\bm q}_0\in\mathbb R^{nL_0}\) be its coefficient vector. Extend
it to \(\bm q\in\mathbb R^{nL}\) by setting
\begin{equation*}
  \bm q=\sqrt{\frac{L}{L_0}}
  (\widetilde{\bm q}_0,\bm 0).
\end{equation*}
Then
\begin{equation}
  \frac{1}{\sqrt{nL}}\mathcal A^*\bm q
  =\frac{1}{\sqrt{nL_0}}\mathcal A_{[L_0]}^*\widetilde{\bm q}_0,
  \qquad
  \|\bm q\|_2
  =\sqrt{\frac{L}{L_0}}\|\widetilde{\bm q}_0\|_2\leq C'',
  \label{cert:eq:certificate-rescaling-properties}
\end{equation}
where \(\mathcal A_{[L_0]}\) denotes the measurement operator restricted to
the first \(L_0\) masks.
Thus, outside an event of probability at most \(n^{-\omega}\), there exist
\(\bm q\in\mathbb R^m\) and
\begin{equation}
  \bm Z=m^{-1/2}\mathcal A^*\bm q,
  \qquad \|\bm q\|_2\leq C'',
  \qquad \|\proj{\tang}\bm Z\|_F\leq\varepsilon,
  \qquad
  \proj{\tangp}\bm Z\preceq-\frac12(\bm I-\bm X_\star).
  \label{cert:eq:certificate-properties}
\end{equation}
Since \(\HT\in\tang\), \(\HTp\in\tangp\), and \(\HTp\succeq\bm0\),
the certificate properties give
\begin{align}
  \langle\bm Z,\bm H\rangle
  &=\langle\proj{\tang}\bm Z,\HT\rangle
    +\langle\proj{\tangp}\bm Z,\HTp\rangle \notag\\
  &\leq \varepsilon\|\HT\|_F
    -\frac12\langle\bm I-\bm X_\star,\HTp\rangle \notag\\
  &=\varepsilon\|\HT\|_F
    -\frac12\operatorname{tr}(\HTp),
  \label{eq:certificate-pairing}
\end{align}
where the last equality uses
\(\langle\bm X_\star,\HTp\rangle=0\). Rearrangement, followed by the adjoint
identity and Cauchy--Schwarz, yields
\begin{align}
  \frac12\operatorname{tr}(\HTp)
  &\leq \bigl|\langle\bm Z,\bm H\rangle\bigr|
       +\varepsilon\|\HT\|_F \notag\\
  &=\frac1{\sqrt m}
    \bigl|\langle\bm q,\mathcal A(\bm H)\rangle\bigr|
       +\varepsilon\|\HT\|_F \notag\\
  &\leq\frac{\|\bm q\|_2}{\sqrt m}\|\mathcal A(\bm H)\|_2
       +\varepsilon\|\HT\|_F.
  \label{eq:Htp-upper}
\end{align}

\paragraph{Erasure-robust tangent estimate}
Under \eqref{cert:eq:stability-assumptions} with \(C\geq2\) required by \cref{prob:eq-a3-assumptions} in \cref{thm:tangent-lower-outside-large},
\cref{thm:tangent-lower-outside-large} holds with failure probability at
most \(2\eta+C_Mn e^{-cL}\).  Apply it with \(\bm P=\HTp\succeq\bm 0\),
where \(E_{\HTp}\) and \(\zeta_{\HTp}\) are defined in
\eqref{eq:defofzetaandE}.  Then
\begin{equation}
  \|\mathcal A_{E_{\HTp}^c}(\HT)\|_2 / \sqrt m
  \geq\frac1{\sqrt8}\|\HT\|_F,
  \qquad
  \zeta_{\HTp} \lesssim\Gamma_{n,m}.
  \label{cert:eq:tangent-estimate}
\end{equation}
Thus, the tangent injectivity bound remains valid after removing
\(E_{\HTp}\), while every remaining index satisfies
\(\langle\bm A_i,\HTp\rangle
\leq\zeta_{\HTp}\operatorname{tr}(\HTp)\).

\paragraph{Combining the estimates}
Applied with \(\bm P=\HTp\),
\cref{lem:psd-l2-after-row-removal} gives
\begin{equation}
  \|\mathcal A_{E_{\HTp}^c}(\HTp)\|_2 / \sqrt{m}
  \leq C_M'\sqrt{\Gamma_{n,m}}\operatorname{tr}(\HTp).
  \label{eq:Htp-l2-bound}
\end{equation}
Since \(\bm H=\HT+\HTp\), linearity gives
\begin{equation}
  \mathcal A_{E_{\HTp}^c}(\HT)
  =\mathcal A_{E_{\HTp}^c}(\bm H)
   -\mathcal A_{E_{\HTp}^c}(\HTp).
\end{equation}
Therefore, the triangle inequality and monotonicity of the \(\ell_2\)-norm
under restriction imply
\begin{align}
  \|\mathcal A_{E_{\HTp}^c}(\HT)\|_2
  &\leq
  \|\mathcal A_{E_{\HTp}^c}(\bm H)\|_2
  +\|\mathcal A_{E_{\HTp}^c}(\HTp)\|_2 \notag\\
  &\leq\|\mathcal A(\bm H)\|_2
  +\|\mathcal A_{E_{\HTp}^c}(\HTp)\|_2.
  \label{cert:eq:restricted-triangle}
\end{align}
Dividing \eqref{cert:eq:restricted-triangle} by \(\sqrt m\) and using
\eqref{cert:eq:tangent-estimate} and \eqref{eq:Htp-l2-bound} in
the resulting inequality gives
\begin{align}
  \frac1{\sqrt8}\|\HT\|_F
  &\leq\frac{\|\mathcal A(\bm H)\|_2}{\sqrt m}
       + C_M'\sqrt{\Gamma_{n,m}}\operatorname{tr}(\HTp).
  \label{cert:eq:tangent-before-certificate}
\end{align}

Substituting \eqref{eq:Htp-upper} into
\eqref{cert:eq:tangent-before-certificate} and using
\(\|\bm q\|_2\leq C''\), we obtain
\begin{equation}
  \frac1{\sqrt8}\|\HT\|_F
  \leq C_M''\sqrt{\Gamma_{n,m}}\left(
       \frac{\|\mathcal A(\bm H)\|_2}{\sqrt m}
       +\varepsilon\|\HT\|_F\right).
  \label{cert:eq:tangent-absorption}
\end{equation}
Choose \(c_0\) in \eqref{eq:eps} sufficiently small so that the last term is
at most \(\|\HT\|_F/(2\sqrt8)\).  Absorbing this term into the left-hand side
gives
\begin{equation}
  \|\HT\|_F
  \lesssim \sqrt{\Gamma_{n,m}}\,
       \frac{\|\mathcal A(\bm H)\|_2}{\sqrt m}.
  \label{cert:eq:tangent-by-residual}
\end{equation}
Substituting \eqref{cert:eq:tangent-by-residual} into
\eqref{eq:Htp-upper} gives
\begin{equation}
  \operatorname{tr}(\HTp)
  \lesssim \frac{\|\mathcal A(\bm H)\|_2}{\sqrt m}.
  \label{cert:eq:normal-by-residual}
\end{equation}
Finally, \(\|\HTp\|_F\leq\operatorname{tr}(\HTp)\) and
\(\HT\perp\HTp\), so
\begin{equation*}
  \|\bm H\|_F^2
  \lesssim \frac{\Gamma_{n,m}}{m}\|\mathcal A(\bm H)\|_2^2.
\end{equation*}
By \cref{cert:eq:log-comparison},
\(\Gamma_{n,m}\lesssim_{M,\omega}\log^4(2n)\), so the preceding estimate
proves \eqref{eq:measurement-stability}.
Moreover, \(L\geq C_1\log n\) gives
\(n e^{-cL}\leq n^{-\omega}\) after increasing \(C_1\). Hence the probability
in \eqref{cert:eq:stability-probability} follows by a union bound over the
tangent and certificate events.
\end{proof}

\section{Proof of the Uniform Row-Subset Bound}
\label{sec:probabilistic}

This section proves the uniform row-subset operator norm bound stated in
\cref{prob:prop-fixed-row-subset}. The main challenge is the dependence among
rows generated by the same mask. We address this dependence by expressing the
row-subset operator norm as a quadratic chaos process in the independent mask
entries.

Recall that \(\bm F=[\bm f_0,\ldots,\bm f_{n-1}]\) is the unnormalized
Fourier matrix.  Thus,
\begin{equation}
  \bm F^*\bm F=\bm F\bm F^*=n\bm I,
  \qquad |[\bm F]_{jk}|=1.
  \label{prob:eq-fourier}
\end{equation}
The measurement matrix \(\bm\Phi\in\mathbb C^{m\times n}\), where \(m=nL\),
was defined in \eqref{prob:eq-analysis-matrix} by
\((\bm\Phi\bm h)_{\ell k}=\bm a_{\ell k}^*\bm h\), with \(\bm a_{\ell k}=\bm D_\ell^* \bm f_k\).  

\subsection{The Row-subset Bound as a Chaos Process}

For \(E\subset[m]\), let \(\bm S_E\) be the coordinate restriction onto
\(E\), and set \(\bm\Phi_E=\bm S_E\bm\Phi\).
Using adjointness and the variational characterization of the operator
norm, we obtain
\begin{align}
\sup_{\substack{E\subset[m]\\ |E|\le s}}
\|\bm\Phi_E\|_{\mathrm{op}}^2
&=
\sup_{\substack{E\subset[m]\\ |E|\le s}}
\|\bm\Phi^*\bm S_E^*\|_{\mathrm{op}}^2
\nonumber\\
&=
\sup_{\substack{E\subset[m]\\ |E|\le s}}
\sup_{\|\bm w\|_2=1}
\|\bm\Phi^*\bm S_E^*\bm w\|_2^2
\nonumber\\
&=
\sup_{\substack{\bm z\in\mathbb C^m\\
\|\bm z\|_2=1,\ \|\bm z\|_0\le s}}
\|\bm\Phi^*\bm z\|_2^2.
\label{prob:eq-row-subset-duality}
\end{align}
Indeed, \(\bm S_E^*\) lifts a vector in \(\mathbb C^{|E|}\) through its
zero-padded extension to \(\mathbb C^m\), and every \(s\)-sparse vector in \(\mathbb C^{m}\) can be
obtained in this way.  
Hence, uniformly bounding the row-subset operator norms is equivalent to
uniformly bounding \(\|\bm\Phi^*\bm z\|_2\) over \(s\)-sparse unit vectors
\(\bm z\in\mathbb C^m\).

Write
\(\bm z=(\bm z_1^T,\ldots,\bm z_L^T)^T\in\mathbb C^{nL}\), with
\(\bm z_\ell\in\mathbb C^n\), and stack the conjugated mask entries as
\begin{equation}
  \bm\xi=(\overline d_{1,0},\ldots,\overline d_{1,n-1},
           \ldots,
           \overline d_{L,0},\ldots,\overline d_{L,n-1})^T \in\mathbb C^{nL}.
  \label{prob:eq-mask-vector}
\end{equation}

Since $\bm\Phi^*\bm z$ depends linearly on $\bm\xi$, it can be expressed as a matrix-vector product $\bm A_{\bm z}\bm\xi$, where $\bm A_{\bm z}\in\mathbb C^{n\times nL}$ depends on $\bm z$.
A direct
calculation gives
\begin{align}
  \bm\Phi^*\bm z
  &=\sum_{\ell=1}^L\sum_{k=0}^{n-1}
       z_{\ell k}\bm D_\ell^*\bm f_k
    =\sum_{\ell=1}^L
       \operatorname{diag}(\bm F\bm z_\ell)\overline{\bm d}_\ell
    =\bm A_{\bm z}\bm\xi.
  \label{prob:eq-chaos-representation}
\end{align}
Thus all dependence among the Fourier measurements generated by one mask is
encoded in the matrix class \(\{\bm A_{\bm z}\}\), whereas the coordinates
of \(\bm\xi\) are independent.
The deterministic matrix associated with \(\bm z\) can be expressed in block form as
\begin{equation}
  \bm A_{\bm z}
  =\big[\operatorname{diag}(\bm F\bm z_1),\ldots,
        \operatorname{diag}(\bm F\bm z_L)\big]
  \in\mathbb C^{n\times nL}.
  \label{prob:eq-chaos-matrix}
\end{equation}
Combining \eqref{prob:eq-row-subset-duality} and
\eqref{prob:eq-chaos-representation} gives
\begin{equation} \label{eq:chaos_representation}
  \sup_{\substack{E\subset[m]\\ |E|\leq s}}
  \|\bm\Phi_E\|_{\mathrm{op}}^2
  =\sup_{\substack{\bm z\in\mathbb C^m\\
                    \|\bm z\|_2\leq1,\ \|\bm z\|_0\leq s}}
   \|\bm A_{\bm z}\bm\xi\|_2^2.
\end{equation}
Thus, the desired row-subset bound reduces to controlling a quadratic chaos
process indexed by sparse vectors. The next subsection establishes the
required uniform estimate.

\subsection{A Uniform Chaos Estimate}

We use a bound for suprema of chaos processes proposed by Krahmer et
al.~\cite[Theorem~3.1]{krahmer2014suprema}.
Their generic-chaining bound was
developed to establish restricted isometries for structured random matrices,
including partial random circulant and time-frequency matrices.  Although we
do not study the restricted isometry property of the measurement matrix
\(\bm\Phi\) itself, \eqref{prob:eq-row-subset-duality} shows that bounding its
row-subset operator norm \(\|\bm\Phi_E\|_{\mathrm{op}}\) reduces to establishing a
one-sided restricted isometry bound for \(\bm\Phi^*\) over sparse vectors.
Thus, the Krahmer--Mendelson--Rauhut theorem provides the appropriate tool for
this intermediate step.

\begin{theorem}[Krahmer--Mendelson--Rauhut chaos bound]
\label{prob:lem-kmr}
Let \(\bm\xi=(\xi_j)_{j=1}^q\) have independent complex coordinates with
\(\mathbb E\xi_j=0\), \(\mathbb E|\xi_j|^2=1\), and a common
\(K\)-sub-Gaussian bound.
Let \(\mathcal B\) be a bounded class of complex
matrices with \(q\) columns, and set
\begin{align}
  d_F&=\sup_{\bm B\in\mathcal B}\|\bm B\|_F,
  &d_{\mathrm{op}}&=\sup_{\bm B\in\mathcal B}\|\bm B\|_{\mathrm{op}},
  \label{prob:eq-kmr-radii}\\[-1mm]
  E_{\mathcal B}
  &=\gamma_2(\mathcal B,\|\cdot\|_{\mathrm{op}})
    \big(\gamma_2(\mathcal B,\|\cdot\|_{\mathrm{op}})+d_F\big)
    +d_Fd_{\mathrm{op}},
  \label{prob:eq-kmr-mean-parameter}\\
  V_{\mathcal B}
  &=d_{\mathrm{op}}\big(\gamma_2(\mathcal B,\|\cdot\|_{\mathrm{op}})+d_F\big),
  \qquad
  U_{\mathcal B}=d_{\mathrm{op}}^2.
  \label{prob:eq-kmr-parameters}
\end{align}
Then, for every \(t>0\),
\begin{equation}
  \mathbb P\left\{
    \sup_{\bm B\in\mathcal B}
    \left|\|\bm B\bm\xi\|_2^2
      -\mathbb E\|\bm B\bm\xi\|_2^2\right|
    \ge C_K E_{\mathcal B}+t
  \right\}
  \le 2\exp\left[-c_K\min\left\{
       \frac{t^2}{V_{\mathcal B}^2},
       \frac{t}{U_{\mathcal B}}
  \right\}\right].
  \label{prob:eq-kmr-tail}
\end{equation}
\end{theorem}
Here, \(\gamma_2(\mathcal B,\|\cdot\|_{\mathrm{op}})\) denotes the
\(\gamma_2\)-functional of the metric space
\((\mathcal B,\|\cdot\|_{\mathrm{op}})\).  Dudley's entropy integral bounds
it in terms of the covering numbers of \(\mathcal B\); see~\cite[(2.1), p.~1883]{krahmer2014suprema}:
\begin{equation}
  \gamma_2(\mathcal B,\|\cdot\|_{\mathrm{op}})
  \lesssim \int_0^{\sup_{\bm B\in\mathcal B}\|\bm B\|_{\mathrm{op}}}
     \sqrt{\log\mathcal N(
       \mathcal B,\|\cdot\|_{\mathrm{op}},u)}\,du.
  \label{prob:eq-dudley}
\end{equation}

Under \cref{ass:mask}, the coordinates
of \(\bm\xi\) in \eqref{prob:eq-mask-vector} are independent, centered,
have unit second moment, and are bounded by \(M\).  
They therefore satisfy
the theorem with a sub-Gaussian parameter depending only on \(M\).

For \(1\le s\le n\), introduce
\begin{equation}
  \mathcal D_s=
  \{\bm z\in\mathbb C^m:\|\bm z\|_2\le1,\ \|\bm z\|_0\le s\},
  \qquad
  \mathcal B_s=\{\bm A_{\bm z}:\bm z\in\mathcal D_s\}.
  \label{prob:eq-sparse-class}
\end{equation}
Given the class \(\mathcal B_s\), we next estimate the geometric parameters needed in \cref{prob:lem-kmr}.  First, the
block form in \eqref{prob:eq-chaos-matrix} gives
\begin{align}
  \bm A_{\bm z}\bm A_{\bm z}^*
  &=\sum_{\ell=1}^L
    \operatorname{diag}(\bm F\bm z_\ell)
    \operatorname{diag}(\bm F\bm z_\ell)^*
  \nonumber\\
  &=\operatorname{diag}\left(
      \sum_{\ell=1}^L|[\bm F\bm z_\ell]_0|^2,\ldots,
      \sum_{\ell=1}^L|[\bm F\bm z_\ell]_{n-1}|^2
    \right).
  \label{prob:eq-block-row-gram}
\end{align}
The block structure and \eqref{prob:eq-fourier} give the Frobenius
radius exactly:
\begin{align}
  \|\bm A_{\bm z}\|_F^2
  &=\operatorname{tr}(\bm A_{\bm z}\bm A_{\bm z}^*)
    =\sum_{\ell=1}^L\|\bm F\bm z_\ell\|_2^2
  \nonumber\\
  &=n\sum_{\ell=1}^L\|\bm z_\ell\|_2^2
    =n\|\bm z\|_2^2.
  \label{prob:eq-frobenius-radius}
\end{align}
For the operator radius $d_{\mathrm{op}}$, every Fourier entry has unit modulus, so
\(|[\bm F\bm z_\ell]_j|\leq\|\bm z_\ell\|_1\). Combining this bound with
\(\|\bm A_{\bm z}\|_{\mathrm{op}}^2
=\|\bm A_{\bm z}\bm A_{\bm z}^*\|_{\mathrm{op}}\) and
\eqref{prob:eq-block-row-gram} gives
\begin{align}
  \|\bm A_{\bm z}\|_{\mathrm{op}}^2
  &=\max_{0\le j<n}\sum_{\ell=1}^L
       |[\bm F\bm z_\ell]_j|^2
  \le\sum_{\ell=1}^L\|\bm z_\ell\|_1^2
  \nonumber\\
  &\le\left(\sum_{\ell=1}^L\|\bm z_\ell\|_1\right)^2
    =\|\bm z\|_1^2
  \le s\|\bm z\|_2^2.
  \label{prob:eq-radii-verification}
\end{align}
Consequently, taking the suprema over \(\bm z\in\mathcal D_s\) gives
\begin{equation}
  d_F(\mathcal B_s)=\sqrt n,
  \qquad
  d_{\mathrm{op}}(\mathcal B_s)\le\sqrt s.
  \label{prob:eq-radii-values}
\end{equation}

It remains to control the \(\gamma_2\)-functional.  Applying
\eqref{prob:eq-dudley} to \(\mathcal B_s\) and using
\eqref{prob:eq-radii-values} gives
\begin{equation}
  \gamma_2(\mathcal B_s,\|\cdot\|_{\mathrm{op}})
  \lesssim \int_0^{\sqrt{s}}
     \sqrt{\log\mathcal N(
       \mathcal B_s,\|\cdot\|_{\mathrm{op}},u)}\,du.
  \label{prob:eq-gamma-two-dudley}
\end{equation}

We estimate the covering numbers in this integral at small and large scales
separately.
For any \(\bm z,\bm z'\in\mathcal D_s\), the vector
\(\bm z-\bm z'\) is supported on at most \(2s\) coordinates. 
By linearity of the map \(\bm z\mapsto\bm A_{\bm z}\),
\(\bm A_{\bm z}-\bm A_{\bm z'}=\bm A_{\bm z-\bm z'}\). Applying
\eqref{prob:eq-radii-verification} to the \(2s\)-sparse vector
\(\bm z-\bm z'\) gives
\begin{equation}
  \|\bm A_{\bm z}-\bm A_{\bm z'}\|_{\mathrm{op}}
  \le\sqrt{2s}\|\bm z-\bm z'\|_2.
  \label{prob:eq-fixed-support-metric}
\end{equation}
We now cover \(\mathcal B_s\) directly.  For \(\rho>0\), let
\(\mathcal N_\rho\) be the union of Euclidean \(\rho\)-nets of the unit
balls on all coordinate subsets of \([m]\) of cardinality at most \(s\).
Every vector in \(\mathcal D_s\) belongs to at least one such ball.  Since
each ball has real dimension at most \(2s\),
\begin{equation}
  |\mathcal N_\rho|
  \le \sum_{k=0}^s \binom mk(1+2/\rho)^{2k}
  \le \left(\frac{em}{s}\right)^s(1+2/\rho)^{2s}.
  \label{prob:eq-sparse-net-cardinality}
\end{equation}
Every \(\bm z\in\mathcal D_s\) has a \(\bm z'\in\mathcal N_\rho\)
with \(\|\bm z-\bm z'\|_2\le\rho\).  Thus,
\eqref{prob:eq-fixed-support-metric} shows that
\(\{\bm A_{\bm z}:\bm z\in\mathcal N_\rho\}\) is a
\(\sqrt{2s}\rho\)-net of \(\mathcal B_s\) in the operator norm.  Taking
\(\rho=u/\sqrt{2s}\) in \eqref{prob:eq-sparse-net-cardinality} gives
\begin{equation}
  \log\mathcal N(\mathcal B_s,\|\cdot\|_{\mathrm{op}},u)
  \le s\log\frac{em}{s}
      +2s\log\left(1+\frac{C\sqrt s}{u}\right).
  \label{prob:eq-fine-entropy}
\end{equation}
Integrating \eqref{prob:eq-fine-entropy} directly over the full range in
\eqref{prob:eq-gamma-two-dudley} yields only a bound of order \(\widetilde{O}(s)\).
 For larger scales, we instead use the 
normed-space version of Maurey's
empirical method~\cite[Lem.~4.2]{krahmer2014suprema}.

\begin{lemma}[Maurey's empirical method]
\label{prob:lem-maurey}
There exists an absolute constant \(C>0\) with the following property.  Let
\((X,\|\cdot\|_X)\) be a normed space and let \(\mathcal U\subset X\) be
finite.  Suppose that, for some \(A>0\), every \(q\in\mathbb N\) and every
\(\bm V_1,\ldots,\bm V_q\in\mathcal U\) satisfy
\begin{equation}
  \mathbb E_\varepsilon
  \left\|\sum_{a=1}^q\varepsilon_a\bm V_a\right\|_X
  \le A\sqrt q,
  \label{prob:eq-maurey-hypothesis}
\end{equation}
where \(\varepsilon_1,\ldots,\varepsilon_q\) are independent Rademacher
variables.  Then, for every \(v>0\),
\begin{equation}
  \log\mathcal N(\operatorname{conv}(\mathcal U),\|\cdot\|_X,v)
  \le C\left(\frac{A}{v}\right)^2\log|\mathcal U|.
  \label{prob:eq-maurey-general}
\end{equation}
\end{lemma}

We verify the hypothesis of \cref{prob:lem-maurey} for the operator norm.
Let
\(\bm U_{\ell k}=\bm A_{\bm e_{\ell k}}\), where
\(\bm e_{\ell k}\) is a coordinate vector in \(\mathbb C^m\).  Each
\(\bm U_{\ell k}\) has one nonzero block,
\(\operatorname{diag}(\bm f_k)\), and has operator norm one.  
Consider the four-phase coordinate atom set
\begin{equation}
  \mathcal U=
  \{\pm\bm U_{\ell k},\ \pm i\bm U_{\ell k}:
      1\le\ell\le L,\ 0\le k<n\},
  \qquad |\mathcal U|=4m.
  \label{prob:eq-complex-atoms}
\end{equation}
The set we are interested in is contained in a multiple of the convex hull of these atoms:
\begin{equation}
  \mathcal B_s\subset\sqrt{2s}\operatorname{conv}(\mathcal U).
  \label{prob:eq-sparse-atomic-inclusion}
\end{equation}
Indeed, linearity gives
\(\bm A_{\bm z}=\sum_{\ell,k}z_{\ell k}\bm U_{\ell k}\).  Writing
\(z_{\ell k}=a_{\ell k}+ib_{\ell k}\), we have
\begin{equation}
  \sum_{\ell,k}(|a_{\ell k}|+|b_{\ell k}|)
  \le\sqrt2\|\bm z\|_1
  \le\sqrt{2s}.
  \label{prob:eq-complex-convex-hull}
\end{equation}
Since \(\bm{0}\in\operatorname{conv}(\mathcal U)\), this proves
\eqref{prob:eq-sparse-atomic-inclusion} by adding the remaining weights on the atom $\bm{0}$.  
It remains to verify the hypothesis of \cref{prob:lem-maurey} for
\(\mathcal U\) under the operator norm and to bound the covering number of
\(\operatorname{conv}(\mathcal U)\), which in turn bounds that of
\(\mathcal B_s\).

Fix \(\bm V_1,\ldots,\bm V_q\in\mathcal U\), and let
\(\varepsilon_1,\ldots,\varepsilon_q\) be independent Rademacher variables.
The operator-norm form of the rectangular noncommutative Khintchine
inequality~\cite[Cor.~4.2]{tropp2012user} gives
\begin{align}
  \mathbb E_\varepsilon
  \left\|\sum_{a=1}^q\varepsilon_a\bm V_a\right\|_{\mathrm{op}}
  &\le C\sqrt{\log(n+m)}
  \max\left\{
    \left\|\sum_{a=1}^q\bm V_a\bm V_a^*\right\|_{\mathrm{op}}^{1/2},
    \left\|\sum_{a=1}^q\bm V_a^*\bm V_a\right\|_{\mathrm{op}}^{1/2}
  \right\}
  \nonumber\\
  &\le C\sqrt{q\log(2m)}.
  \label{prob:eq-atom-expectation}
\end{align}
Indeed, \(\bm V_a\bm V_a^*=\bm I_n\) and
\(\|\bm V_a^*\bm V_a\|_{\mathrm{op}}=1\), so both terms in the maximum are at
most \(\sqrt q\) by direct summation.
Thus the hypothesis of \cref{prob:lem-maurey} holds with
\(A=C\sqrt{\log(2m)}\).  Applying the lemma, using
\(|\mathcal U|=4m\), and then scaling by the factor \(\sqrt{2s}\) in
\eqref{prob:eq-sparse-atomic-inclusion} yields
\begin{align}
  \log\mathcal N(\mathcal B_s,\|\cdot\|_{\mathrm{op}},u)
  &\le \log\mathcal N\left(
    \operatorname{conv}(\mathcal U),\|\cdot\|_{\mathrm{op}},
    \frac{u}{\sqrt{2s}}\right)
  \nonumber
  \le \log(4m) \frac{C\log(2m)}{u^2/(2s)}
  \nonumber\\
  &\le \frac{Cs\log^2(8m)}{u^2},
  \qquad 0<u\le\sqrt{s}.
  \label{prob:eq-maurey-entropy}
\end{align}

Using \eqref{prob:eq-fine-entropy} for \(0<u\le1\) and
\eqref{prob:eq-maurey-entropy} for \(1<u\le\sqrt{s}\) in
\eqref{prob:eq-gamma-two-dudley} gives
\begin{align}
  \int_0^1
    \sqrt{\log\mathcal N(\mathcal B_s,\|\cdot\|_{\mathrm{op}},u)}\,du
  &\le C\sqrt{s\log(2m)},
  \label{prob:eq-small-scale-integral}\\
  \int_1^{\sqrt{s}}
    \sqrt{\log\mathcal N(\mathcal B_s,\|\cdot\|_{\mathrm{op}},u)}\,du
  &\le C\sqrt{s}\log(8m)\log(2s).
  \label{prob:eq-large-scale-integral}
\end{align}
Substituting these estimates into \eqref{prob:eq-gamma-two-dudley} gives
\begin{align}
  \gamma_2(\mathcal B_s,\|\cdot\|_{\mathrm{op}})
  &\le C\sqrt s\left[
       \sqrt{\log(2m)}
       +\log(2s)\log(8m)
       \right]
   \le C\sqrt{s \Gamma_{n,m}},
  \qquad 1\le s\le n.
  \label{prob:eq-gamma-bound}
\end{align}
where we recall that \(\Gamma_{n,m}=\log^2(2m)\log^2(2n)\) and 
\(\Lambda_{n,m,\eta}=\Gamma_{n,m}+\log(2/\eta)\).

Substituting these estimates into \cref{prob:lem-kmr} gives the proof of \cref{prob:prop-fixed-row-subset}.
\begin{proof}[Proof of \Cref{prob:prop-fixed-row-subset}]
For every \(\bm z\), independence and unit second moment of the mask
entries imply
\begin{equation}
  \mathbb E\|\bm A_{\bm z}\bm\xi\|_2^2
  =\|\bm A_{\bm z}\|_F^2=n\|\bm z\|_2^2.
  \label{prob:eq-chaos-center}
\end{equation}
Apply \cref{prob:lem-kmr} to \(\mathcal B_s\).  
From
\cref{prob:eq-radii-values} and \cref{prob:eq-gamma-bound},
\begin{align}
  E_{\mathcal B_s}
  &\le C_M\left(\sqrt{ns\Gamma_{n,m}}+s\Gamma_{n,m}\right),
  V_{\mathcal B_s}
  \le C_M\left(s\sqrt{\Gamma_{n,m}}+\sqrt{ns}\right),
  \qquad U_{\mathcal B_s}\le s.
  \label{prob:eq-kmr-parameters-simplified}
\end{align}
Set
\(t=C_M'(V_{\mathcal B_s}\sqrt{\log(2/\eta)}
+U_{\mathcal B_s}\log(2/\eta))\). Then
\(t^2/V_{\mathcal B_s}^2\geq(C_M')^2\log(2/\eta)\) and
\(t/U_{\mathcal B_s}\geq C_M'\log(2/\eta)\), so a sufficiently large
\(C_M'\) makes the failure probability in \eqref{prob:eq-kmr-tail} at most
\(\eta\). Moreover, \eqref{prob:eq-kmr-parameters-simplified},
\(2\sqrt{\Gamma_{n,m}\log(2/\eta)}\leq\Lambda_{n,m,\eta}\), and
\(\max\{\log(2/\eta),\Gamma_{n,m}\}\leq\Lambda_{n,m,\eta}\) give
\(t\leq C_M''(\sqrt{ns\Lambda_{n,m,\eta}}+s\Lambda_{n,m,\eta})\).
Thus, with probability at least
\(1-\eta\),
\begin{equation}
  \sup_{\bm z\in\mathcal D_s}\|\bm A_{\bm z}\bm\xi\|_2^2
  \le n+C_M\left(\sqrt{ns \Lambda_{n,m,\eta}}+s \Lambda_{n,m,\eta}\right),
  \label{prob:eq-fixed-row-chaos-bound}
\end{equation}
where we also used \(\|\bm z\|_2\le1\) in
\eqref{prob:eq-chaos-center}.
Finally, \cref{eq:chaos_representation} identifies the $\sup_{\bm z\in\mathcal D_s}\|\bm A_{\bm z}\bm\xi\|_2^2$ with the supremum of
\(\|\bm\Phi_E\|_{\mathrm{op}}^2\) over \(|E|\le s\), proving the proposition.
\end{proof}

\section{Construction of a Dual Certificate}\label{sec:dual_cert_proof}
The stability proof in \cref{sec:mainproof} invoked the pure-range dual
certificate in \cref{prop:pure-range-certificate}; we now give its construction
and the deferred proof of the proposition.
Existing golfing constructions for the CDP model produce a certificate in $\operatorname{range}(\mathcal A^*)+\operatorname{span}\{\bm I\}$~\cite[Definition~1 and Proposition~7]{huangli2026optimal}.
For our purposes, however, we require the certificate to lie exactly in $\operatorname{range}(\mathcal A^*)$, together with explicit control of its coefficient vector.

The modification is to replace the identity matrix $\bm I$ by an empirical identity that lies in $\operatorname{range}(\mathcal A^*)$. The following proposition shows that this replacement incurs only a small perturbation, sufficient to preserve the estimates used in the subsequent golfing construction.

\begin{proposition}[Empirical identity and its PSD consequence]
\label{prop:empirical-identity}
Under \cref{ass:mask}, for every $0<\delta<1$, with
probability at least $1-2n\exp(-\frac{2L\delta^2}{M^4})$,
\begin{equation}
  (1-\delta)\bm I\preceq
  \frac{\mathcal A^*(\mathbf 1)}{m}
  \preceq(1+\delta)\bm I.
  \label{eq:prelim-empirical-identity-order}
\end{equation}
On the same event, simultaneously for every $\bm P\succeq0$,
\begin{equation}
  (1-\delta)\operatorname{tr}(\bm P)
  \leq\frac{\|\mathcal A(\bm P)\|_1}{m}
  \leq(1+\delta)\operatorname{tr}(\bm P).
  \label{eq:prelim-empirical-identity-psd}
\end{equation}
\end{proposition}
\begin{proof}
Equation~\eqref{eq:prelim-empirical-identity-order} follows  from~\cite[Lemma~3.3]{candes2015cdp}.
Since $\bm P \succeq 0$ implies $[\mathcal A(\bm P)]_{\ell k}\geq0$ for all $k,l$, $m^{-1}\|\mathcal A(\bm P)\|_1 = \langle m^{-1}\mathcal A^*(\mathbf 1), \bm P \rangle$, and taking the inner product of \eqref{eq:prelim-empirical-identity-order} with $\bm P$ immediately yields \eqref{eq:prelim-empirical-identity-psd}.
\end{proof}

We now develop the remaining ingredients of the construction.
Inspired by the stage-wise construction in~\cite{huangli2026optimal}, we introduce truncation events and their
associated truncated measurement operators. Let \(B\) be a batch
consisting of $L_B$
masks. For a fixed nonzero matrix
\(\bm Q\in\mathcal T_\star\) and a truncation parameter \(\tau\geq1\),
define
\begin{align}
 U_{\ell k}^{\tau}(\bm Q)
 &:=
 \left\{
   |\langle\bm A_{\ell k},\bm Q\rangle|
   \leq K_0M^2\tau\|\bm Q\|_F
 \right\},
 \label{cert:eq:trunc-event}\\
 \mathcal R_{B,\tau}^{\bm Q}(\bm W)
 &:=
 \frac{1}{nL_B}
 \sum_{\ell\in B}\sum_{k=0}^{n-1}
 \mathbf 1_{U_{\ell k}^{\tau}(\bm Q)}
 \bm A_{\ell k}
 \langle\bm A_{\ell k},\bm W\rangle.
 \label{cert:eq:truncated-operator}
\end{align}
Here, \(K_0>0\) is a sufficiently large numerical constant.

The following lemma summarizes the results of~\cite[Propositions~5 and~6]{huangli2026optimal}, which establish estimates for the $\mathcal T_\star$ and $\mathcal T_\star^\perp$ components of these adaptive truncated operators.

\begin{lemma}
\label{cert:lem:fresh-batch}
There are numerical constants $C_0,c_1,c_2>0$ such that the following holds, provided the numerical cutoff constant $K_0$ in \eqref{cert:eq:trunc-event} is sufficiently large.  Let $\bm{Q}\in\mathcal T_\star$ be fixed independently of a fresh batch $B$ of $L_B$ masks.  If $\tau\geq1$, $t_1>0$, and $0<t_2\leq1$ obey
\begin{equation}
 C_0M^4e^{-\tau}\leq\frac14\min\{t_1,t_2\},
 \label{cert:eq:bias-condition}
\end{equation}
then
\begin{align}
 &\mathbb P\!\left\{
  \left\|P_{\mathcal T_\star^\perp}
   \big(\mathcal R_{B,\tau}^{\bm{Q}}(\bm{Q})-\operatorname{tr}(\bm{Q})\bm{I}\big)  \right\|_{\mathrm{op}}>t_1\|\bm{Q}\|_F\right\}
 \nonumber\\[-1mm]
 &\hspace{35mm}\leq
 2n\exp\!\left[-\frac{c_1L_B}{M^8\tau}
                   \min\{t_1^2,t_1\}\right],
 \label{cert:eq:fresh-normal}\\
 &\mathbb P\!\left\{
  \left\|P_{\mathcal T_\star}
   \big(\mathcal R_{B,\tau}^{\bm{Q}}(\bm{Q})-\bm{Q}-\operatorname{tr}(\bm{Q})\bm{I}\big)
  \right\|_F>t_2\|\bm{Q}\|_F\right\}
 \nonumber\\[-1mm]
 &\hspace{35mm}\leq
 \exp\!\left[-\frac{c_2t_2^2L_B}{M^8\tau}+\frac14\right].
 \label{cert:eq:fresh-tangent}
\end{align}
The same bounds hold conditionally when $\bm{Q}$ is measurable with respect to
the past and $B$ is independent of that past.
\end{lemma}

\begin{proof}
See the proofs of Propositions~5 and~6 in~\cite{huangli2026optimal}.
\end{proof}

We construct the required dual certificate using an adaptive golfing scheme with \(L=O(\log n + \log(1/\epsilon))\) masks. Unlike the PhaseLift construction, we require an exact range certificate of the form \(\mathcal A^*(\bm q)\), together with explicit control of \(\|\bm q\|_2\).

\begin{proof}[Proof of Proposition \ref{prop:pure-range-certificate}]
We first construct an empirical identity $\bm S_0\in\operatorname{range}(\mathcal A^*)$ and use it to initialize an adaptive golfing scheme. 
Using independent fresh batches, we construct
updates $\bm U_j\in\operatorname{range}(\mathcal A^*)$ that geometrically contract the tangent residual $\bm Q_j$, while keeping the accumulated $\tangp$ component small. 
We then control the total number of masks required by the adaptive trials. 
Finally, we verify the desired $\tang$ and $\tangp$ bounds for $\bm Z=(nL)^{-1/2}\mathcal{A}^*(\bm q)$ and estimate the norm of its coefficient vector $\bm q$.
\paragraph{Parameter Selection}
Set
\begin{equation}
 \delta:=\frac1{32},
 \qquad
 \eta:=\frac1{32},
 \qquad
 t_2:=\frac14,
 \qquad
 \rho:=t_2+\sqrt2\,\delta<\frac13.
 \label{cert:eq:golf-parameters}
\end{equation}
These choices ensure that
\begin{equation}
 (1+\sqrt2\,\delta)
 \left(
  \frac{\eta}{1-\sqrt\rho}
  +\frac{\delta}{1-\rho}
 \right)
 <\frac14.
 \label{cert:eq:parameter-margin}
\end{equation}

Choose
\begin{equation}
 J
 :=
 \left\lceil
 \frac{\log\bigl((1+\sqrt2\,\delta)/\varepsilon\bigr)}
      {|\log\rho|}
 \right\rceil,
 \label{cert:eq:number-levels}
\end{equation}
to be the number of successful golfing updates. For
\(0\leq j\leq J-1\), define
\begin{equation}
 t_{1,j}:=\eta\rho^{-j/2},
 \qquad
 a_j:=\min\{t_{1,j}^2,t_{1,j}\}.
 \label{cert:eq:level-parameters}
\end{equation}

Fix \(K_0\) sufficiently large for
\Cref{cert:lem:fresh-batch}, and choose
$\tau_0
 :=
 \max\left\{
  1,\log(128C_0M^4)
 \right\}.
$
Since \(t_{1,j}\geq\eta\), this choice guarantees, at every level $j$, that
\[
 C_0M^4e^{-\tau_0}
 \leq
 \frac14\min\{t_{1,j},t_2\}.
\]

For \(0\leq j\leq J-1\) and a sufficiently large constant \(C_1\), let $C_M:=C_1M^8\tau_0$ and
\begin{equation}
 L_j
 :=
 \left\lceil
 C_M\left(1+\frac{\log n}{a_j}\right)
 \right\rceil,
 \label{cert:eq:batch-size}
\end{equation}
be the number of masks used in each trial at level \(j\). Finally, let
\begin{equation}
 s_0
 :=
 \left\lceil
 2M^4\omega\delta^{-2}\log n
 \right\rceil,
 \label{cert:eq:identity-batch-size}
\end{equation}
be the number of masks used in the empirical-identity block.
\paragraph{Empirical-Identity Initialization}
We construct $\bm S_0$, which will be used to initialize the golfing scheme.
Let \(\mathcal C\) be a batch of \(s_0\) masks, disjoint from all batches
used in the golfing updates, and define
\begin{equation}
 \bm S_0
 :=\frac{1}{ns_0}\mathcal A_{\mathcal C}^*(\bm 1)
 =\frac{1}{s_0}\sum_{\ell\in\mathcal C}
   \bm D_\ell^*\bm D_\ell.
 \label{cert:eq:empirical-identity}
\end{equation}
Thus, \(\bm S_0\) is diagonal,
\(\bm S_0\in\operatorname{range}(\mathcal A_{\mathcal C}^*)\), and
\(\mathbb E\bm S_0=\bm I\). By
\cref{prop:empirical-identity},
\begin{equation}
 \|\bm S_0-\bm I\|_{\mathrm{op}}\leq\delta,
 \label{cert:eq:empirical-identity-event}
\end{equation}
holds with probability at least $1-2n\exp\!\left(-{2s_0\delta^2}/{M^4}\right)$.
On the event
\(\|\bm S_0-\bm I\|_{\mathrm{op}}\leq\delta\),
\begin{equation}
 \|P_{\mathcal T_\star}(\bm S_0-\bm I)\|_F
 \leq\sqrt{2}\,\delta,
 \qquad
 \|P_{\mathcal T_\star^\perp}(\bm S_0-\bm I)\|_{\mathrm{op}}
 \leq\delta.
 \label{cert:eq:identity-projections}
\end{equation}

\paragraph{Adaptive Golfing Construction} We now introduce the adaptive golfing construction for the dual certificate.
At level \(j\), repeatedly draw a fresh batch of \(L_j\) masks until
both conclusions of \cref{cert:lem:fresh-batch} hold with
\(\bm Q=\bm Q_j\), \(\tau=\tau_0\), \(t_1=t_{1,j}\), and
\(t_2=1/4\). Denote the accepted batch by \(B_j\) and write
\[
 \mathcal R_j^{\bm Q_j}
 :=
 \mathcal R_{B_j,\tau_0}^{\bm Q_j}.
\]
Thus, the accepted operator satisfies
\begin{align}
 \left\|
 \proj{\tangp}
 \bigl(
   \mathcal R_j^{\bm Q_j}(\bm Q_j)
   -\operatorname{tr}(\bm Q_j)\bm I
 \bigr)
 \right\|_{\mathrm{op}}
 &\leq t_{1,j}\|\bm Q_j\|_F,
 \label{cert:eq:accepted-normal}\\
 \left\|
 \proj{\tang}
 \bigl(
   \mathcal R_j^{\bm Q_j}(\bm Q_j)
   -\bm Q_j-\operatorname{tr}(\bm Q_j)\bm I
 \bigr)
 \right\|_F
 &\leq t_2\|\bm Q_j\|_F.
 \label{cert:eq:accepted-tangent}
\end{align}
A rejected batch is assigned zero coefficient. If \(\bm Q_j=\bm 0\),
the construction stops and all subsequent increments are set to zero.

Initialize
\begin{equation}
 \bm{G}_0=\bm{0},
 \qquad \bm{Q}_0=P_{\mathcal T_\star}\bm{S}_0.
 \label{cert:eq:initial-residual}
\end{equation}
For every accepted level, define
\begin{align}
 \bm U_j
 &:=
 \mathcal R_j^{\bm Q_j}(\bm Q_j)
 -\operatorname{tr}(\bm Q_j)\bm S_0,
 \label{cert:eq:pure-update}\\
 \bm G_{j+1}
 &:=
 \bm G_j+\bm U_j,
 \qquad
 \bm Q_{j+1}
 :=
 \bm Q_j-\proj{\tang}(\bm U_j).
 \label{cert:eq:residual-update}
\end{align}
We replace the ideal term $\operatorname{tr}(\bm{Q}_j)\bm{I}$ by $\operatorname{tr}(\bm{Q}_j)\bm{S}_0$ to ensure both terms in \eqref{cert:eq:pure-update} belong to \(\operatorname{range}(\mathcal A^*)\). From the above construction, we also have $\bm{U}_j\in\operatorname{range}(\mathcal A^*)$ and $\bm{Q}_j\in\mathcal T_\star$.

After $J$ accepted levels, set
\begin{equation}
 \bm{Z}:=\bm{G}_J-\bm{S}_0.
 \label{cert:eq:local-certificate}
\end{equation}
Consequently,
\(\bm Z\in\operatorname{range}(\mathcal A^*)\).

\paragraph{Success Probability and Total Mask Count}
Let \(\mathcal F_{\mathrm{past}}\) denote the randomness generated by all batches used before the current trial at level \(j\). Conditional on \(\mathcal F_{\mathrm{past}}\), the residual
\(\bm Q_j\in\tang\) is fixed and independent of the fresh trial batch.
By \eqref{cert:eq:fresh-normal} and
\eqref{cert:eq:fresh-tangent}, the conditional failure probability of
the current trial is bounded above by
\begin{align}
2n\exp\left(
 -\frac{c_1L_ja_j}{M^8\tau_0}
\right)
+
\exp\left(
 -\frac{c_2t_2^2L_j}{M^8\tau_0}+\frac14
\right).
\label{cert:eq:trial-failure-initial}
\end{align}
Since \(C_M=C_1M^8\tau_0\), the definition of \(L_j\) in \eqref{cert:eq:batch-size}
gives
\begin{equation}
 \frac{L_ja_j}{M^8\tau_0}
 \geq C_1(a_j+\log n)
 \geq C_1\log n,
 \qquad
 \frac{L_j}{M^8\tau_0}
 \geq C_1.
 \label{cert:eq:batch-exponents}
\end{equation}

Hence, for a sufficiently large constant \(C_1\), there exists an absolute constant $\alpha_0$ such that
\begin{equation}
 \mathbb P\left\{
  \text{the current trial fails}
  \,\middle|\,\mathcal F_{\mathrm{past}}
 \right\}
 \leq
 \exp\left(
  -\frac{\alpha_0L_j}{M^8\tau_0}
 \right)
 =:p_j.
 \label{cert:eq:trial-failure}
\end{equation}

Let \(N_j\) be the number of trials required to obtain the successful
update at level \(j\). For every integer \(r\geq0\),
\begin{equation}
 \mathbb P\left\{
  N_j>r
  \,\middle|\,
  N_0,\ldots,N_{j-1}
 \right\}
 \leq p_j^r.
 \label{cert:eq:geometric-domination}
\end{equation}
Thus \(N_j\) is conditionally dominated by a geometric random variable
with success probability \(1-p_j\). The total number of masks consumed
by the golfing updates is
\begin{equation}
 L_{\mathrm{golf}}
 :=
 \sum_{j=0}^{J-1}L_jN_j.
 \label{cert:eq:golf-mask-consumption}
\end{equation}

Set $B:=\sum_{j=0}^{J-1}L_j$. By \eqref{cert:eq:level-parameters}, \(a_0=\eta^2\), and thus
\begin{equation}
 B\geq L_0
 \geq C_M\eta^{-2}\log n\gtrsim_M\log n.
 \label{cert:eq:baseline-lower}
\end{equation}

On the other hand, by \eqref{cert:eq:level-parameters},
\begin{equation}
\frac1{a_j}
\leq
 \frac1{t_{1,j}^2}+\frac1{t_{1,j}}
 =
 \eta^{-2}\rho^j+\eta^{-1}\rho^{j/2}.
\end{equation}
Therefore,
\begin{equation}
\sum_{j=0}^{J-1}\frac1{a_j}\leq \sum_{j=0}^{J-1}( \eta^{-2}\rho^j+\eta^{-1}\rho^{j/2})\leq   \frac{\eta^{-2}}{1-\rho}
   +\frac{\eta^{-1}}{1-\sqrt\rho}.
\end{equation}

Equation~\eqref{cert:eq:batch-size} and the definition of \(J\) in \eqref{cert:eq:number-levels} give
\begin{align}
 B
 &\leq
 2C_M\left(
  J+\log n\sum_{j=0}^{J-1}\frac1{a_j}
 \right)
\lesssim_M\log n+\log(1/\varepsilon).
 \label{cert:eq:budget-sum}
\end{align}

Set $\theta:=\frac{\alpha_0}{2M^8\tau_0}$. By \eqref{cert:eq:trial-failure}, \(p_j=e^{-2\theta L_j}\). Then the conditional moment
generating function obeys
\begin{equation}
 \mathbb E\left[
  e^{\theta L_jN_j}
  \,\middle|\,
  N_0,\ldots,N_{j-1}
 \right]
 \leq
 \frac{(1-p_j)e^{\theta L_j}}
      {1-p_je^{\theta L_j}}
 =
 1+e^{\theta L_j}
 \leq e^{2\theta L_j}.
 \label{cert:eq:geometric-mgf}
\end{equation}
For $\omega\geq1$, iterated conditional expectation
and Markov's inequality yield
\begin{equation}
 \mathbb P\left\{
  L_{\mathrm{golf}}>4\omega B
 \right\}
 \leq
 e^{-4\omega\theta B}
 \prod_{j=0}^{J-1}e^{2\theta L_j}
 \leq
 \exp\left(
  -\frac{\alpha_0\omega B}{M^8\tau_0}
 \right).
 \label{cert:eq:golf-mask-tail}
\end{equation}
By \eqref{cert:eq:batch-size} and \(a_0=\eta^2\),
\begin{align}
 \frac{\alpha_0B}{M^8\tau_0}
 &\geq
 \frac{\alpha_0L_0}{M^8\tau_0}
 \geq
 \alpha_0C_1\eta^{-2}\log n
 \geq\log(2n).
 \label{cert:eq:baseline-exponent}
\end{align}
Consequently,
\begin{equation}
 \mathbb P\left\{
  L_{\mathrm{golf}}>4\omega B
 \right\}
 \leq
 (2n)^{-\omega}
 \leq\frac12n^{-\omega}.
 \label{cert:eq:golf-exhaustion-probability}
\end{equation}

For the empirical-identity block,
by the choice of $s_0$ in \eqref{cert:eq:identity-batch-size} and \cref{prop:empirical-identity},
\begin{equation}
 \mathbb P\left\{
  \|\bm S_0-\bm I\|_{\mathrm{op}}>\delta
 \right\}
 \leq
 2n\exp\left(
  -{2s_0\delta^2}/{M^4}
 \right)
 \leq
 2n^{1-4\omega}
 \leq\frac12n^{-\omega}.
 \label{cert:eq:identity-failure-probability}
\end{equation}

Therefore, it suffices to reserve $s_0+\left\lceil4\omega B\right\rceil$ masks. Let the total number $L$ of available masks satisfy
\begin{equation}
 L\gtrsim_M
 \omega\bigl(\log n+\log(1/\varepsilon)\bigr).
 \label{cert:eq:total-mask-reservation}
\end{equation}
By \eqref{cert:eq:identity-batch-size} and \eqref{cert:eq:budget-sum}, a union bound applied to \eqref{cert:eq:golf-exhaustion-probability} and \eqref{cert:eq:identity-failure-probability} shows that the empirical-identity event holds and all
\(J\) successful golfing updates are obtained before the available
masks are exhausted with probability at least $1-n^{-\omega}$.

\paragraph{Dual Certificate Verification}
We now verify that on the intersection of the empirical-identity event
\eqref{cert:eq:empirical-identity-event} and the accepted-trial events
\eqref{cert:eq:accepted-normal}--\eqref{cert:eq:accepted-tangent}, the resulting $\bm{Z}$ is the desired dual certificate.

Let
\begin{equation}
 \bm{E}_{\mathcal T,j}:=P_{\mathcal T_\star}
 \big(\mathcal R_j^{\bm{Q}_j}(\bm{Q}_j)-\bm{Q}_j-\operatorname{tr}(\bm{Q}_j)\bm{I}\big).
 \label{cert:eq:tangent-error}
\end{equation}
Using $\bm{Q}_j\in\mathcal T_\star$, the residual update is the exact
identity
\begin{equation}
 \bm{Q}_{j+1}=-\bm{E}_{\mathcal T,j}
  +\operatorname{tr}(\bm{Q}_j)P_{\mathcal T_\star}(\bm{S}_0-\bm{I}).
 \label{cert:eq:exact-recurrence}
\end{equation}
Since $
|\operatorname{tr}(\bm Q_j)|
 \leq\|\bm Q_j\|_F$ for any $\bm Q_j\in\tang$, \eqref{cert:eq:empirical-identity-event} and \eqref{cert:eq:identity-projections} give
 \begin{equation}
  \|\operatorname{tr}(\bm{Q}_j)P_{\mathcal T_\star}(\bm{S}_0-\bm{I})\|_F\leq \sqrt{2}\delta\|\bm{Q}_j\|_F,\quad \|\operatorname{tr}(\bm{Q}_j)P_{\mathcal T_\star^\perp}(\bm{S}_0-\bm{I})\|_{\mathrm{op}}\leq \delta\|\bm{Q}_j\|_F.   
 \end{equation}
Then \eqref{cert:eq:accepted-normal} and 
\eqref{cert:eq:accepted-tangent} imply that each successful update satisfies
\begin{equation}
\|\bm{Q}_{j+1}\|_F\leq\rho\|\bm{Q}_j\|_F,
 \qquad
 \|P_{\mathcal T_\star^\perp}\bm{U}_j\|_{\mathrm{op}}
 \leq(t_{1,j}+\delta)\|\bm{Q}_j\|_F.
 \label{cert:eq:contraction-leakage}
\end{equation}

By the update formula, 
$$\bm{G}_J=\sum_{j=0}^{J-1}\bm{U}_j.$$

For the $\mathcal T_\star$ component,
\begin{equation}
 P_{\mathcal T_\star}\bm{Z}
 =\sum_{j=0}^{J-1}P_{\mathcal T_\star}\bm{U}_j
   -P_{\mathcal T_\star}\bm{S}_0
 =(\bm{Q}_0-\bm{Q}_J)-\bm{Q}_0=-\bm{Q}_J.
 \label{cert:eq:exact-telescope}
\end{equation}
Moreover, \eqref{cert:eq:identity-projections} gives 
\begin{equation}\label{cert:eq:Q0-bound}
\|\bm{Q}_0\|_F\leq1+\sqrt2\,\delta.    
\end{equation}
Thus, with the choice of $J$ in \eqref{cert:eq:number-levels},
\begin{equation}
 \|P_{\mathcal T_\star}\bm{Z}\|_F
 \leq\rho^J\|\bm{Q}_0\|_F
 \leq(1+\sqrt2\,\delta)\rho^J\leq\varepsilon.
 \label{cert:eq:tangent-accuracy}
\end{equation}

For the $\mathcal T_\star^\perp$ component, $t_{1,j}=\eta\rho^{-j/2}$, \eqref{cert:eq:contraction-leakage} and \eqref{cert:eq:Q0-bound} give
\begin{equation}
\|P_{\mathcal T_\star^\perp}\bm{U}_j\|_{\mathrm{op}}
\leq(\eta\rho^{-j/2}+\delta)\rho^j\|\bm{Q}_0\|_F
\leq (1+\sqrt2\,\delta)\big(\eta\rho^{j/2}+\delta\rho^j\big).  
\end{equation}
Therefore, choosing $\delta$ and $\eta$ so that \eqref{cert:eq:parameter-margin} holds,
\begin{align}
 \|P_{\mathcal T_\star^\perp}\bm{G}_J\|_{\mathrm{op}}
 &\leq
 \sum_{j=0}^{J-1}\|P_{\mathcal T_\star^\perp}\bm{U}_j\|_{\mathrm{op}}
 \nonumber\\
 &\leq(1+\sqrt2\,\delta)
 \sum_{j=0}^{J-1}\big(\eta\rho^{j/2}+\delta\rho^j\big)
 \nonumber\\
 &\leq(1+\sqrt2\,\delta)
 \left(\frac{\eta}{1-\sqrt\rho}+\frac{\delta}{1-\rho}\right)
 \leq \frac{1}{4}.
 \label{cert:eq:normal-leakage-sum}
\end{align}
Then $
\proj{\tangp}\bm G_J
\preceq
\frac{1}{4}
(\bm I-\bm X_\star)$.
Furthermore, \eqref{cert:eq:empirical-identity-event} implies
$P_{\mathcal T_\star^\perp}\bm{S}_0\succeq
(1-\delta)(\bm{I}-\bm{X}_\star)\succeq
\frac{3}{4}(\bm{I}-\bm{X}_\star)$. Consequently,
\begin{equation}
 P_{\mathcal T_\star^\perp}\bm{Z}
 \preceq-\frac12(\bm{I}-\bm{X}_\star).
 \label{cert:eq:local-normal-margin}
\end{equation}

It remains to control the coefficient vector.
By the update formulas \eqref{cert:eq:pure-update} and
\eqref{cert:eq:residual-update}, we have
\begin{equation}
\bm{Z}=\sum_{j=0}^{J-1}\bm{U}_j-\bm{S}_0=\sum_{j=0}^{J-1}\mathcal R_j^{\bm{Q}_j}(\bm{Q}_j)-\Bigl(1+\sum_{j=0}^{J-1}\operatorname{tr}(\bm{Q}_j)\Bigr)\bm{S}_0.
\end{equation}
Since $\bm{Z}\in\operatorname{range}(\mathcal A^*)$, the dual certificate has an
exact representation
\begin{equation}
 \bm{Z}=\mathcal A^*
 \bm{\lambda},
 \label{cert:eq:local-range}
\end{equation}
for a coefficient vector \(\bm\lambda\in\mathbb R^{nL}\). Since the accepted batches and the empirical-identity batch are
disjoint, we consider them separately.

On the accepted batch \(B_j\), the coefficients associated with $\mathcal R_j^{\bm{Q}_j}(\bm{Q}_j)$ are
\begin{equation}
 \lambda_{\ell k}^{(j)}
 =\frac1{nL_j}\mathbf{1}_{U_{\ell k}^{\tau_0}(\bm{Q}_j)}
   \langle \bm{A}_{\ell k},\bm{Q}_j\rangle,
   \qquad
 \ell\in B_j,\quad 0\leq k\leq n-1.
 \label{cert:eq:golf-coefficients}
\end{equation}
Since $\bm{Q}_j\in\mathcal T_\star$, $|\operatorname{tr}(\bm Q_j)| \leq \|\bm Q_j\|_F$ and then by \eqref{cert:eq:accepted-tangent},
\begin{equation}\label{cert:eq-tangent-RQj}
\left\|
\proj{\tang}
\bigl(\mathcal R_j^{\bm Q_j}(\bm Q_j)\bigr)
\right\|_F
\leq
 \left\|
 \proj{\tang}
 \bigl(
   \bm Q_j+\operatorname{tr}(\bm Q_j)\bm I
 \bigr)
 \right\|_F
+t_2\|\bm Q_j\|_F\leq (2+t_2)\|\bm Q_j\|_F.   
\end{equation}
Further, $\langle \bm{Q}_j,\mathcal R_j^{\bm{Q}_j}(\bm{Q}_j)\rangle\leq\|\bm{Q}_j\|_F\|\proj{\tang}
\bigl(\mathcal R_j^{\bm Q_j}(\bm Q_j)\bigr)\|_F$.
Then by \eqref{cert:eq:contraction-leakage} and \eqref{cert:eq-tangent-RQj},
\begin{equation}
\|\bm{\lambda}^{(j)}\|_2^2
 =\frac1{nL_j}\langle \bm{Q}_j,\mathcal R_j^{\bm{Q}_j}(\bm{Q}_j)\rangle
 \leq
 \frac{2+t_2}{nL_j}\|\bm{Q}_j\|_F^2
 \leq\frac{(2+t_2)(1+\sqrt2\,\delta)^2\rho^{2j}}{nL_j}.
 \label{cert:eq:update-coefficient-bound}   
\end{equation}
By \eqref{cert:eq:level-parameters} and \eqref{cert:eq:batch-size}, $a_j\leq t_{1,j}^2\leq\eta^2\rho^{-j}$ and $L_j
 \geq
 C_M\frac{\log n}{a_j}$.
Therefore,
\begin{equation}
 \sum_{j=0}^{J-1}\|\bm\lambda^{(j)}\|_2^2
 \leq
 \frac{(2+t_2)(1+\sqrt2\,\delta)^2}
      {C_Mn\log n}
 \sum_{j=0}^{J-1}\rho^{2j}a_j
\leq
 \frac{(2+t_2)(1+\sqrt2\,\delta)^2\eta^2}
      {C_M(1-\rho)n\log n}.
 \label{cert:eq:update-coefficient-sum}
\end{equation}

The coefficients from $\bm{S}_0$ are all equal to
\begin{equation}
 -\frac{\vartheta}{ns_0},
 \qquad
 \vartheta:=1+\sum_{j=0}^{J-1}\operatorname{tr}(\bm{Q}_j).
 \label{cert:eq:identity-coefficients}
\end{equation}
Moreover,
\begin{equation}
 |\vartheta|
 \leq
 1+\sum_{j=0}^{J-1}\|\bm Q_j\|_F
 \leq
 1+\frac{1+\sqrt2\,\delta}{1-\rho}:=C_{\vartheta}.
 \label{cert:eq:theta-bound}
\end{equation}
Using \eqref{cert:eq:identity-batch-size}, the squared norm of these coefficients satisfies
\begin{equation}
 ns_0\left|\frac{\vartheta}{ns_0}\right|^2
 \leq
 \frac{C_\vartheta^2\delta^2}
      {2M^4\omega\,n\log n}.
 \label{cert:eq:identity-coefficient-norm}
\end{equation}

All failed and unused batches receive coefficient zero. Since the supports are disjoint, 
we obtain
\begin{equation}
 \|\bm\lambda\|_2^2
 =
 \sum_{j=0}^{J-1}\|\bm\lambda^{(j)}\|_2^2
 +\frac{|\vartheta|^2}{ns_0}
 \leq
 \frac{C_{\lambda,M,\omega}}{n\log n},
 \label{cert:eq:local-coefficient-norm}  
\end{equation}
where
\begin{equation}
 C_{\lambda,M,\omega}
 :=
 \frac{(2+t_2)(1+\sqrt2\,\delta)^2\eta^2}
      {C_M(1-\rho)}
 +
 \frac{C_\vartheta^2\delta^2}{2M^4\omega}.
 \label{cert:eq:coefficient-constant}
\end{equation}

Finally, define $\bm q:=\sqrt{nL}\,\bm\lambda$; then
$\bm Z=(nL)^{-1/2}\mathcal A^*(\bm q)$. 
By \eqref{cert:eq:total-mask-reservation}, when we choose 
$L =\lceil C_{M,\omega}(\log n+\log(1/\varepsilon))\rceil$ for a sufficiently large constant \(C_{M,\omega}\), we have
\begin{equation}
 \|\bm q\|_2^2
 =nL\|\bm\lambda\|_2^2
 \leq C_{\lambda,M,\omega}\frac{L}{\log n}
\lesssim_{M,\omega}
 \left(
  1+\frac{\log(1/\varepsilon)}{\log n}
 \right).
 \label{cert:eq:q-norm-bound}
\end{equation}

Combining the preceding estimates \eqref{cert:eq:tangent-accuracy} \eqref{cert:eq:local-normal-margin} and \eqref{cert:eq:q-norm-bound}, the matrix $\bm Z$ constructed in \eqref{cert:eq:local-certificate} is the desired pure-range dual certificate, which
completes the proof of \cref{prop:pure-range-certificate}.
\end{proof}

\section{From Stability to a Benign Landscape}
\label{sec:deterministic-landscape}

In this section, we prove \cref{cor:benign-landscape-noisy} by combining the deterministic landscape framework of~\cite{mcrae2026benign} with \cref{thm:stability}.
We show that every real second-order critical point of \eqref{eq:factorized-objective} satisfies a recovery bound under noise. In the noiseless case, this implies a benign strict-saddle landscape: every second-order critical point is a global minimizer, and every nonglobal critical point has a direction of negative real-Hessian curvature.

\subsection{First- and Second-Order Critical Conditions for the Nonconvex Loss}
Consider the nonconvex least-squares problem
\eqref{eq:factorized-objective}. For simplicity, we first consider the case where the ground truth $\bm X_\star$ and the factor $\bm V$ are real matrices, and the measurement operator $\mathcal{A}$ is a real linear operator. In this case, the gradient of $F_r$ and
the Hessian of $F_r$ along a direction $\dot{\bm V}\in\mathbb R^{n\times r}$ are both well-defined and have the following formulas:
\begin{equation}
  \nabla F_r(\bm V)
  =\frac{1}{nL}\mathcal A^*(\mathcal A
    (\bm V\bm V^*)-\bm y)\bm V.
  \label{eq:factorized-first-derivative}
\end{equation}
\begin{align}
  \nabla^2F_r(\bm V)[\dot{\bm{V}},\dot{\bm{V}}]
  =\frac{1}{nL}\bigg(&
    \left\langle
      \mathcal A(\bm V\bm V^*)-\bm y,
      \mathcal A( \dot{ \bm V}\dot{ \bm V}^*)
    \right\rangle
    \nonumber\\
    &+\frac12
    \left\|\mathcal A(\bm V \dot{\bm V}^*+ \dot{\bm V}\bm V^*)\right\|_2^2
  \bigg).
  \label{eq:factorized-second-derivative}
\end{align}
$\bm V$ is a real second-order critical point if and only if
it satisfies the conditions
\begin{equation}
    \nabla F_r(\bm V)=\bm 0,\quad
    \nabla^2F_r(\bm V)[\dot{\bm{V}},\dot{\bm{V}}]\geq 0
    \quad\text{for every }\dot{\bm V}\in\mathbb R^{n\times r}.
    \label{eq:factorized-second-order-critical-point}
\end{equation}
The landscape is benign if every real second-order critical point $\bm V$ satisfies $\bm{VV^*} = \bm{X_\star}$.

We now discuss the complex case, which is the main focus of this paper. We rewrite the problem in real variables by embedding the measurement operator $\mathcal{A}$ and the optimization variable $\bm V$ into the real domain, using the same convention as in~\cite[Section~3]{mcrae2026benign}:
\begin{align}
  &\bm A = \bm A_{\mathcal{R}} + i \bm A_{\mathcal{I}} \mapsto \tilde{\bm A}=\begin{bmatrix} \bm{A}_{\mathcal{R}} & \bm{A}_{\mathcal{I}}^\top  \\ \bm{A}_{\mathcal{I}} & \bm{A}_{\mathcal{R}} \end{bmatrix}\in \mathbb{R}^{2n \times 2n} , \\ &\bm V = \bm{V}_{\mathcal{R}}+i\bm{V}_{\mathcal{I}}\mapsto \tilde{\bm V}=\begin{bmatrix} \bm{V}_{\mathcal{R}} \\ \bm{V}_{\mathcal{I}} \end{bmatrix}\in \mathbb{R}^{2n \times r}.
\end{align}
For each Hermitian matrix $\bm A$, one can verify that $\langle \bm A, \bm{VV^*}\rangle = \langle \tilde{\bm A}, \tilde{\bm V}\tilde{\bm V}^\top \rangle$. Define the transformed operator $\tilde{\mathcal{A}}$ by $[\tilde{\mathcal{A}}(\bm{X})]_j=\langle \tilde{\bm A_j}, \bm{X}\rangle$. The complex problem \eqref{eq:factorized-objective} can thus be rewritten as the following real optimization problem, with $\tilde{\bm X_\star}:=\tilde{\bm x_\star}\tilde{\bm x_\star}^\top$ and $\bm y=\tilde{\mathcal{A}}(\tilde{\bm{X}_\star})+\bm e$:
\begin{equation}
    \min_{\tilde{\bm V}\in \mathbb{R}^{2n\times r}} \tilde{F}_r(\tilde{\bm V}) = \frac{1}{4nL}\|\tilde{\mathcal{A}}(\tilde{\bm V}\tilde{\bm V}^\top )-\bm y\|_2^2.
    \label{eq:factorized-objective-real}
\end{equation}
In this context, a second-order critical point of the complex problem \eqref{eq:factorized-objective} is defined as the complex preimage of a second-order critical point of the real problem \eqref{eq:factorized-objective-real}; that is, $\bm{V}$ is a second-order critical point of \eqref{eq:factorized-objective} if and only if $\tilde{\bm V}$ satisfies
\begin{equation}
    \nabla \tilde{F}_r(\tilde{\bm V})=\bm 0,\quad
    \nabla^2\tilde{F}_r(\tilde{\bm V})[\dot{\tilde{\bm{V}}},\dot{\tilde{\bm{V}}}]\geq 0
    \quad\text{for every }\dot{\tilde{\bm V}}\in\mathbb R^{2n\times r}.
    \label{eq:factorized-second-order-critical-point-real}
\end{equation}

\subsection{The Consequence of the Critical Conditions and its Relationship to Stability}
We then explain how the critical conditions and overparameterization yield an error bound for second-order critical points.
We begin by testing the second-order condition in
\cref{eq:factorized-second-derivative} along rank-one directions
$
\dot{\bm V}:=\bm u\bm v_j^*,
$
where $\bm u:=\bm x_\star-\bm V\bm q$ for arbitrary $\bm q$ and \(\{\bm v_j\}_{j=1}^r\) is an orthonormal basis. 
Summing the resulting inequalities and combining them with the first-order condition in \cref{eq:factorized-first-derivative} bounds the squared measurement residual
by a noise term and a second term proportional to \(1/r\); see~\cite{mcrae2026benign} for the detailed proof.

\begin{lemma}[Lemma~1 of~\cite{mcrae2026benign}, rank-one CDP form]
\label{lem:critical-point-residual}
Let $\bm y=\mathcal A(\bm X_\star)+\bm e$ with
$\bm e\in\mathbb R^{nL}$, and consider
\[
  \bm V\longmapsto\frac{1}{4nL}
  \|\mathcal A(\bm V\bm V^*)-\bm y\|_2^2,
  \qquad \bm V\in\mathbb C^{n\times r}.
\]
If $\bm V$ is a second-order critical point of this loss, then
for every $\bm q\in\mathbb C^r$,
\begin{align}
  \frac{1}{nL}\|\mathcal A(\bm V\bm V^*-\bm X_\star)\|_2^2
  &\leq \frac{1}{m}
  \langle\bm e,\mathcal A(\bm V\bm V^*-\bm X_\star)\rangle
  +\frac{2}{(r+2)m}
  \left\langle\bm y,
    \mathcal A\bigl((\bm x_\star-\bm V\bm q)
                       (\bm x_\star-\bm V\bm q)^*\bigr)\right\rangle
  \notag\\
  &\leq \frac{1}{m}
  \langle\bm e,\mathcal A(\bm V\bm V^*-\bm X_\star)\rangle
  +\frac{2}{(r+2)m}\|\mathcal A^*(\bm y)\|_{\mathrm{op}}
    \|\bm x_\star-\bm V\bm q\|_2^2.
  \label{eq:critical-point-residual}
\end{align}
\end{lemma}

To illustrate its implications, we consider the noiseless case $\bm e = 0$. 
The quantity $\frac{1}{m}\|\mathcal{A}^*(\bm y)\|_{\mathrm{op}}$ reduces to $\frac{1}{m}\|\mathcal{A}^*\mathcal{A}(\bm X_\star)\|_{\mathrm{op}} \lesssim \log(2n)$, which is the spectral upper bound at the ground truth; see \cref{prop:fixed-target-spectral-upper} below.
On the other hand, we can see the essential role of \cref{thm:stability}:
it shows that $\mathcal{A}$ is stably injective in the form
$$
\frac{1}{m}\|\mathcal{A}(\bm{VV^*} - \bm X_\star)\|_2^2 \ge \alpha\|\bm{VV^*} - \bm X_\star\|_F^2, 
$$
where the stability constant $\alpha >0$ determines the loss incurred when transferring the upper bound in measurement space to the matrix space.
Dividing the right-hand side of \cref{eq:critical-point-residual} by $\alpha$, we see that a factor width satisfying
\begin{equation} \label{eq:width-requirement}
r \gtrsim \frac{\frac{1}{m} \|\mathcal{A}^*\mathcal{A}(\bm X_\star)\|_{\mathrm{op}}}{\alpha} 
\end{equation}
is needed to make $\|\bm{VV^*} - \bm X_\star\|_F$ diminish.
This is precisely the order of
overparameterization required by the deterministic framework; see~\cite[Theorem~2]{mcrae2026benign}.
To conclude this subsection, we give the spectral upper bound at $\bm X_\star$ for the CDP model as a direct consequence of the coherence bound in \Cref{prob:lem-coherence}.
\begin{proposition}[Fixed-target logarithmic spectral upper bound]
\label{prop:fixed-target-spectral-upper}
Let $\bm x_\star$ be a deterministic unit vector independent of
the masks, and suppose that $L\le C(M, \omega)\log n$ for some constant $C(M, \omega)$ depending on $M$ and $\omega$. For every $\omega\geq1$,
there exists a constant $C_{M, \omega}>0$ such that
\begin{equation}
\frac1m\mathcal A^*\mathcal A(\bm X_\star)
\preceq
C_{M, \omega}\log(2n)\,\bm I,
\label{eq:prelim-fixed-target-spectral-upper}
\end{equation}
with probability at least $1-n^{-\omega}$.
\end{proposition}
\begin{proof}
Set $\eta=n^{-\omega}$. On the coherent event in
\Cref{prob:lem-coherence}, using
$\sum_{k=0}^{n-1}\bm f_k\bm f_k^*=n\bm I$, we have
\begin{align}
\frac1m\mathcal A^*\mathcal A(\bm X_\star)
&=
\frac1{nL}\sum_{\ell=1}^L\sum_{k=0}^{n-1}
|\bm a_{\ell k}^*\bm x_\star|^2
\bm a_{\ell k}\bm a_{\ell k}^*
\notag\\
&\preceq
\mu_\star\frac1{nL}
\sum_{\ell=1}^L\sum_{k=0}^{n-1}
\bm a_{\ell k}\bm a_{\ell k}^*
\notag\\
&=
\mu_\star\frac1L\sum_{\ell=1}^L
\bm D_\ell^*\bm D_\ell
\preceq
M^2\mustar\,\bm I
\notag\\
&=
4M^4\log\!\left(4n^{\omega+1}L\right)\bm I.
\label{eq:coherence-spectral-upper}
\end{align}
If $L=O(\log n)$, then, with $C_{M, \omega}:=C_M(\omega+1)$, the last quantity is bounded by
\[
C_{M,\omega}\log(2n)\,\bm I.
\]
The coherence event has probability at least
$1-\eta=1-n^{-\omega}$.
\end{proof}

\subsection{Noisy Error Bound and the Benign Noiseless Landscape}
In this subsection, we combine \cref{lem:critical-point-residual} with inputs from Theorem~\ref{thm:stability} and Proposition~\ref{prop:fixed-target-spectral-upper}, as indicated
by \cref{eq:width-requirement}, to determine the required factor width in \cref{cor:benign-landscape-noisy}.
\begin{proof}[Proof of \Cref{cor:benign-landscape-noisy}]
Let \(\bm V\) be a second-order critical point of
\eqref{eq:factorized-objective}, and set
\(\bm H=\bm V\bm V^*-\bm X_\star\). Multiplying
\eqref{eq:critical-point-residual} by \(m=nL\) gives, for every
\(\bm q\in\mathbb C^r\),
\begin{equation}
  \|\mathcal A(\bm H)\|_2^2
  \leq \langle\bm e,\mathcal A(\bm H)\rangle
  +\frac{2}{r+2}\|\mathcal A^*(\bm y)\|_{\mathrm{op}}
    \|\bm x_\star-\bm V\bm q\|_2^2.
\end{equation}

Choose \(\bm q\) to minimize \(\|\bm x_\star-\bm V\bm q\|_2\).
Then \(\bm x_\star-\bm V\bm q\) is orthogonal to the column space of
\(\bm V\). Since \(\|\bm x_\star\|_2=1\), projection onto its orthogonal
complement gives
\begin{equation}
  \|\bm x_\star-\bm V\bm q\|_2
  =\|\proj{\bm range(V)^\perp}\bm X_\star\|_F
  \leq\|\bm H\|_F.
  \label{eq:tangent_and_difference}
\end{equation}
On the event in \cref{thm:stability},
\begin{equation}
  \|\mathcal A(\bm H)\|_2^2
  \geq\frac{mc}{\log^4(2n)}\|\bm H\|_F^2.
\end{equation}
Consequently, for \(\tau\) defined in
\cref{cor:benign-landscape-noisy},
\begin{equation}
  \frac{2}{r+2}\|\mathcal A^*(\bm y)\|_{\mathrm{op}}
    \|\bm x_\star-\bm V\bm q\|_2^2
  \leq\frac{\tau+2}{r+2}\|\mathcal A(\bm H)\|_2^2.
  \label{eq:second-term-smaller-than-zero}
\end{equation}
Substitution and Cauchy--Schwarz in the measurement space yield
\begin{equation}
  \frac{r-\tau}{r+2}\|\mathcal A(\bm H)\|_2^2
  \leq\langle\bm e,\mathcal A(\bm H)\rangle
  \leq\|\bm e\|_2\|\mathcal A(\bm H)\|_2.
  \label{eq:rearranged-critical-condition}
\end{equation}
If \(\mathcal A(\bm H)=\bm 0\), the lower isometry gives \(\bm H=\bm 0\).
Otherwise, canceling one factor of \(\|\mathcal A(\bm H)\|_2\) and applying
the lower isometry again gives
\begin{equation}
  \|\bm H\|_F
  \leq\frac{r+2}{r-\tau}
  \frac{\log^2(2n)}{\sqrt{mc}}\,\|\bm e\|_2,
  \label{eq:error-related-measurement-bound}
\end{equation}
which proves \eqref{eq:noisy-nonconvex-errorbound}.

It remains to specialize the width condition to noiseless measurements. If
\(\bm e=\bm 0\), then \(\bm y=\mathcal A(\bm X_\star)\), and
\cref{prop:fixed-target-spectral-upper} gives
\begin{equation}
  \|\mathcal A^*(\bm y)\|_{\mathrm{op}} = \|\mathcal A^*\mathcal{A}(\bm X_\star)\|_{\mathrm{op}}
  \leq mC_{M,\omega}\log(2n).
\end{equation}
Note that the event of \cref{thm:stability} is included in this event, as the coherence upper bound from \cref{prob:lem-coherence} is invoked in its proof. 
Hence,
\begin{equation}
  \tau
  \leq \frac{2C_{M, \omega}}{c}\log^5(2n)-2
  \leq \frac{C_{M,\omega}}{c}\log^5(2n).
\end{equation}
Thus \eqref{eq:landscape-width} ensures \(r>\max\{\tau,0\}\); setting
\(\bm e=\bm 0\) in \eqref{eq:noisy-nonconvex-errorbound} gives
\(\bm V\bm V^*=\bm X_\star\). By the inclusion above, the intersection of
the stability and spectral events is the stability event, whose complement has
probability at most \(Cn^{-\omega}\) by \eqref{cert:eq:stability-probability}.
\end{proof}

\section{Discussion and Conclusion}
\label{sec:conclusion}

We established stability results at near-optimal scaling for the CDP model at the optimal-order
mask complexity \(L\asymp\log n\). 
For any fixed signal chosen independently of the masks, the lifted measurement
operator satisfies a uniform lower isometry over the PSD secant set. This
stability ensures that PhaseLift-type convex programs achieve stable recovery
with error of order
\(\log^2(2n)\|\bm e\|_2/\sqrt{nL}\). Besides, in the
noiseless case, the nonconvex loss has a benign landscape under the mild
overparameterization \(r=O(\log^5(2n))\).
The key ingredient is a uniform operator-norm bound over row
subsets of the CDP measurement matrix, whose result and proof technique may
be useful in other analyses of the CDP or more structured measurement
models.

Several questions remain open. First, are the logarithmic factors intrinsic,
or can one obtain a dimension-independent lower isometry constant? Such a
constant would sharpen the stable recovery bound and improve the current
logarithmic scalings. Second, does the landscape remain benign at \(r=1\), as
in the Gaussian case~\cite{cai2023nearly}? Finally, the benign landscape result
does not by itself provide an iteration-complexity guarantee. Connecting this
geometric result to quantitative convergence guarantees for practical
phase-retrieval algorithms, such as Wirtinger flow, truncated amplitude flow,
and mirror descent, would complement existing analyses under the CDP
model~\cite{li2022sampling,li2025taf,godeme2023provable}.

\bibliographystyle{IEEEtran}
\bibliography{references}

\begin{thebibliography}{10}
\providecommand{\url}[1]{#1}
\csname url@samestyle\endcsname
\providecommand{\newblock}{\relax}
\providecommand{\bibinfo}[2]{#2}
\providecommand{\BIBentrySTDinterwordspacing}{\spaceskip=0pt\relax}
\providecommand{\BIBentryALTinterwordstretchfactor}{4}
\providecommand{\BIBentryALTinterwordspacing}{\spaceskip=\fontdimen2\font plus
\BIBentryALTinterwordstretchfactor\fontdimen3\font minus
  \fontdimen4\font\relax}
\providecommand{\BIBforeignlanguage}[2]{{%
\expandafter\ifx\csname l@#1\endcsname\relax
\typeout{** WARNING: IEEEtran.bst: No hyphenation pattern has been}%
\typeout{** loaded for the language `#1'. Using the pattern for}%
\typeout{** the default language instead.}%
\else
\language=\csname l@#1\endcsname
\fi
#2}}
\providecommand{\BIBdecl}{\relax}
\BIBdecl

\bibitem{beinert2015ambiguities}
R.~Beinert and G.~Plonka, ``Ambiguities in one-dimensional discrete phase
  retrieval from {Fourier} magnitudes,'' \emph{J. Fourier Anal. Appl.},
  vol.~21, no.~6, pp. 1169--1198, 2015.

\bibitem{bendory2017fourier}
T.~Bendory, R.~Beinert, and Y.~C. Eldar, ``{Fourier} phase retrieval:
  Uniqueness and algorithms,'' in \emph{Compressed Sensing and Its
  Applications}.\hskip 1em plus 0.5em minus 0.4em\relax Cham: Birkh{\"a}user,
  2017, pp. 55--91.

\bibitem{dainty1987astronomy}
J.~C. Dainty and J.~R. Fienup, ``Phase retrieval and image reconstruction for
  astronomy,'' in \emph{Image Recovery: Theory and Application}, H.~Stark,
  Ed.\hskip 1em plus 0.5em minus 0.4em\relax Academic Press, 1987, ch.~7, pp.
  231--275.

\bibitem{harrison1993phase}
R.~W. Harrison, ``Phase problem in crystallography,'' \emph{J. Opt. Soc. Amer.
  A}, vol.~10, no.~5, pp. 1046--1055, 1993.

\bibitem{millane1990phase}
R.~P. Millane, ``Phase retrieval in crystallography and optics,'' \emph{J. Opt.
  Soc. Amer. A}, vol.~7, no.~3, pp. 394--411, 1990.

\bibitem{shechtman2015overview}
Y.~Shechtman, Y.~C. Eldar, O.~Cohen, H.~N. Chapman, J.~Miao, and M.~Segev,
  ``Phase retrieval with application to optical imaging: A contemporary
  overview,'' \emph{IEEE Signal Process. Mag.}, vol.~32, no.~3, pp. 87--109,
  2015.

\bibitem{jaganathan2016overview}
K.~Jaganathan, Y.~C. Eldar, and B.~Hassibi, ``Phase retrieval: An overview of
  recent developments,'' in \emph{Optical Compressive Imaging}, A.~Stern,
  Ed.\hskip 1em plus 0.5em minus 0.4em\relax CRC Press, 2016, pp. 264--292.

\bibitem{candes2013phaselift}
E.~J. Cand{\`e}s, T.~Strohmer, and V.~Voroninski, ``{PhaseLift}: Exact and
  stable signal recovery from magnitude measurements via convex programming,''
  \emph{Commun. Pure Appl. Math.}, vol.~66, no.~8, pp. 1241--1274, 2013.

\bibitem{candes2014quadratic}
E.~J. Cand{\`e}s and X.~Li, ``Solving quadratic equations via {PhaseLift} when
  there are about as many equations as unknowns,'' \emph{Found. Comput. Math.},
  vol.~14, no.~5, pp. 1017--1026, 2014.

\bibitem{candes2015wirtinger}
E.~J. Cand{\`e}s, X.~Li, and M.~Soltanolkotabi, ``Phase retrieval via
  {Wirtinger Flow}: Theory and algorithms,'' \emph{IEEE Trans. Inf. Theory},
  vol.~61, no.~4, pp. 1985--2007, 2015.

\bibitem{chen2017solving}
Y.~Chen and E.~J. Cand{\`e}s, ``Solving random quadratic systems of equations
  is nearly as easy as solving linear systems,'' \emph{Commun. Pure Appl.
  Math.}, vol.~70, no.~5, pp. 822--883, 2017.

\bibitem{chen2019gradient}
Y.~Chen, Y.~Chi, J.~Fan, and C.~Ma, ``Gradient descent with random
  initialization: Fast global convergence for nonconvex phase retrieval,''
  \emph{Math. Program.}, vol. 176, no. 1--2, pp. 5--37, 2019.

\bibitem{duchi2019solving}
J.~C. Duchi and F.~Ruan, ``Solving (most) of a set of quadratic equalities:
  Composite optimization for robust phase retrieval,'' \emph{Inf. Inference},
  vol.~8, no.~3, pp. 471--529, 2019.

\bibitem{godeme2023provable}
J.-J. Godeme, J.~Fadili, X.~Buet, M.~Zerrad, M.~Lequime, and C.~Amra,
  ``Provable phase retrieval with mirror descent,'' \emph{SIAM J. Imaging
  Sci.}, vol.~16, no.~3, pp. 1106--1141, 2023.

\bibitem{tan2023online}
Y.~S. Tan and R.~Vershynin, ``Online stochastic gradient descent with arbitrary
  initialization solves non-smooth, non-convex phase retrieval,'' \emph{J.
  Mach. Learn. Res.}, vol.~24, no.~58, pp. 1--47, 2023.

\bibitem{wang2018truncated}
G.~Wang, G.~B. Giannakis, and Y.~C. Eldar, ``Solving systems of random
  quadratic equations via truncated amplitude flow,'' \emph{IEEE Trans. Inf.
  Theory}, vol.~64, no.~2, pp. 773--794, 2018.

\bibitem{zhang2017nonconvex}
H.~Zhang, Y.~Zhou, Y.~Liang, and Y.~Chi, ``A nonconvex approach for phase
  retrieval: Reshaped {Wirtinger Flow} and incremental algorithms,'' \emph{J.
  Mach. Learn. Res.}, vol.~18, no. 141, pp. 1--35, 2017.

\bibitem{sun2018geometric}
J.~Sun, Q.~Qu, and J.~Wright, ``A geometric analysis of phase retrieval,''
  \emph{Found. Comput. Math.}, vol.~18, no.~5, pp. 1131--1198, 2018.

\bibitem{cai2023nearly}
J.-F. Cai, M.~Huang, D.~Li, and Y.~Wang, ``Nearly optimal bounds for the global
  geometric landscape of phase retrieval,'' \emph{Inverse Problems}, vol.~39,
  no.~7, p. 075011, 2023.

\bibitem{liu2024local}
K.~Liu, Z.~Wang, and L.~Wu, ``The local landscape of phase retrieval under
  limited samples,'' \emph{IEEE Trans. Inf. Theory}, vol.~70, no.~12, pp.
  9012--9035, 2024.

\bibitem{mcrae2026benign}
A.~D. McRae, ``Phase retrieval and matrix sensing via benign and
  overparametrized nonconvex optimization,'' \emph{IEEE Trans. Inf. Theory},
  vol.~72, no.~6, pp. 4203--4220, 2026.

\bibitem{krahmer2020complex}
F.~Krahmer and D.~St{\"o}ger, ``Complex phase retrieval from subgaussian
  measurements,'' \emph{J. Fourier Anal. Appl.}, vol.~26, 2020, art. no. 89.

\bibitem{huangli2026heavy}
G.~Huang and S.~Li, ``Low-rank matrix recovery via heavy-tailed quadratic
  sampling,'' 2026, arXiv:2607.08671.

\bibitem{loewen2018diffraction}
E.~G. Loewen and E.~Popov, \emph{Diffraction Gratings and Applications}.\hskip
  1em plus 0.5em minus 0.4em\relax Boca Raton, FL: CRC Press, 2018.

\bibitem{gerchberg1972practical}
R.~W. Gerchberg and W.~O. Saxton, ``A practical algorithm for the determination
  of phase from image and diffraction plane pictures,'' \emph{Optik}, vol.~35,
  no.~2, pp. 237--246, 1972.

\bibitem{fienup1982phase}
J.~R. Fienup, ``Phase retrieval algorithms: A comparison,'' \emph{Appl. Opt.},
  vol.~21, no.~15, pp. 2758--2769, 1982.

\bibitem{fannjiang2020numerics}
A.~Fannjiang and T.~Strohmer, ``The numerics of phase retrieval,'' \emph{Acta
  Numerica}, vol.~29, pp. 125--228, 2020.

\bibitem{fannjiang2012phase}
A.~Fannjiang and W.~Liao, ``Phase retrieval with random phase illumination,''
  \emph{J. Opt. Soc. Amer. A}, vol.~29, no.~9, pp. 1847--1859, 2012.

\bibitem{candes2015cdp}
E.~J. Cand{\`e}s, X.~Li, and M.~Soltanolkotabi, ``Phase retrieval from coded
  diffraction patterns,'' \emph{Appl. Comput. Harmon. Anal.}, vol.~39, no.~2,
  pp. 277--299, 2015.

\bibitem{gross2017improved}
D.~Gross, F.~Krahmer, and R.~Kueng, ``Improved recovery guarantees for phase
  retrieval from coded diffraction patterns,'' \emph{Appl. Comput. Harmon.
  Anal.}, vol.~42, no.~1, pp. 37--64, 2017.

\bibitem{huangli2026optimal}
G.~Huang and S.~Li, ``{PhaseLift} for coded diffraction patterns: Optimal
  sampling rate,'' 2026, arXiv:2608.02450.

\bibitem{li2021incremental}
J.~Li, J.-F. Cai, and H.~Zhao, ``Scalable incremental nonconvex optimization
  approach for phase retrieval,'' \emph{J. Sci. Comput.}, vol.~87, no.~2,
  p.~43, 2021.

\bibitem{lili2021fourier}
H.~Li and S.~Li, ``Phase retrieval from {Fourier} measurements with masks,''
  \emph{Inverse Probl. Imaging}, vol.~15, no.~5, pp. 1051--1075, 2021.

\bibitem{huangwen2026phaselift}
M.~Huang, J.~Wen, and R.~Zhang, ``Recovery performance of {PhaseLift} for phase
  retrieval from coded diffraction patterns,'' \emph{Inverse Problems},
  vol.~42, p. 045021, 2026.

\bibitem{soltanolkotabi2014algorithms}
M.~Soltanolkotabi, ``Algorithms and theory for clustering and nonconvex
  quadratic programming,'' Ph.D. dissertation, Stanford University, Stanford,
  CA, USA, 2014.

\bibitem{krahmer2014suprema}
F.~Krahmer, S.~Mendelson, and H.~Rauhut, ``Suprema of chaos processes and the
  restricted isometry property,'' \emph{Commun. Pure Appl. Math.}, vol.~67,
  no.~11, pp. 1877--1904, 2014.

\bibitem{recht2010guaranteed}
B.~Recht, M.~Fazel, and P.~A. Parrilo, ``Guaranteed minimum-rank solutions of
  linear matrix equations via nuclear norm minimization,'' \emph{SIAM Review},
  vol.~52, no.~3, pp. 471--501, 2010.

\bibitem{candes2011tight}
E.~J. Cand{\`e}s and Y.~Plan, ``Tight oracle inequalities for low-rank matrix
  recovery from a minimal number of noisy random measurements,'' \emph{IEEE
  Trans. Inf. Theory}, vol.~57, no.~4, pp. 2342--2359, 2011.

\bibitem{jaganathan2015masks}
K.~Jaganathan, Y.~C. Eldar, and B.~Hassibi, ``Phase retrieval with masks using
  convex optimization,'' in \emph{2015 IEEE International Symposium on
  Information Theory (ISIT)}, 2015, pp. 1655--1659.

\bibitem{zhang2021overparameterized}
R.~Y. Zhang, ``Sharp global guarantees for nonconvex low-rank matrix recovery
  in the overparameterized regime,'' 2021.

\bibitem{ma2023overparametrization}
Z.~Ma, I.~Molybog, J.~Lavaei, and S.~Sojoudi, ``Over-parametrization via
  lifting for low-rank matrix sensing: Conversion of spurious solutions to
  strict saddle points,'' in \emph{Proc. 40th Int. Conf. Machine Learning},
  ser. Proc. Mach. Learn. Res., vol. 202, 2023, pp. 23\,373--23\,387.

\bibitem{mcrae2024synchronization}
A.~D. McRae and N.~Boumal, ``Benign landscapes of low-dimensional relaxations
  for orthogonal synchronization on general graphs,'' \emph{SIAM J. Optim.},
  vol.~34, no.~2, pp. 1427--1454, 2024.

\bibitem{nguyen2017loss}
Q.~Nguyen and M.~Hein, ``The loss surface of deep and wide neural networks,''
  in \emph{Proc. 34th Int. Conf. Machine Learning}, ser. Proc. Mach. Learn.
  Res., vol.~70, 2017, pp. 2603--2612.

\bibitem{du2018power}
S.~Du and J.~Lee, ``On the power of over-parametrization in neural networks
  with quadratic activation,'' in \emph{Proc. 35th Int. Conf. Machine
  Learning}, ser. Proc. Mach. Learn. Res., vol.~80, 2018, pp. 1329--1338.

\bibitem{mcrae2026amplitude}
A.~D. McRae, ``Phase retrieval via overparametrized nonconvex optimization:
  Nonsmooth amplitude loss landscapes,'' \emph{IEEE Trans. Inf. Theory},
  vol.~72, no.~9, pp. 7049--7068, 2026.

\bibitem{laska2011democracy}
J.~N. Laska, P.~T. Boufounos, M.~A. Davenport, and R.~G. Baraniuk, ``Democracy
  in action: Quantization, saturation, and compressive sensing,'' \emph{Appl.
  Comput. Harmon. Anal.}, vol.~31, no.~3, pp. 429--443, 2011.

\bibitem{han2017robustness}
B.~Han and Z.~Xu, ``Robustness properties of dimensionality reduction with
  gaussian random matrices,'' \emph{Sci. China Math.}, vol.~60, pp. 1753--1778,
  2017.

\bibitem{dirksen2021nonGaussian}
S.~Dirksen and S.~Mendelson, ``Non-gaussian hyperplane tessellations and robust
  one-bit compressed sensing,'' \emph{J. Eur. Math. Soc.}, vol.~23, no.~9, pp.
  2913--2947, 2021.

\bibitem{chenYuanOneBitPhaseRetrieval}
J.~Chen and M.~Yuan, ``One-bit phase retrieval: Optimal rates and efficient
  algorithms,'' \emph{IEEE Trans. Inf. Theory}, vol.~72, no.~7, pp. 5251--5292,
  2026.

\bibitem{tropp2012user}
J.~A. Tropp, ``User-friendly tail bounds for sums of random matrices,''
  \emph{Found. Comput. Math.}, vol.~12, no.~4, pp. 389--434, 2012.

\bibitem{li2022sampling}
H.~Li, S.~Li, and Y.~Xia, ``Sampling complexity on phase retrieval from masked
  {Fourier} measurements via {Wirtinger Flow},'' \emph{Inverse Problems},
  vol.~38, no.~10, p. 105004, 2022.

\bibitem{li2025taf}
H.~Li and J.~Li, ``Truncated amplitude flow with coded diffraction patterns,''
  \emph{Inverse Problems}, vol.~41, no.~1, p. 015002, 2025.

\end{thebibliography}
\end{document}